\documentclass[preprint,12pt]{elsarticle}

\usepackage{amssymb}
\usepackage{amsmath}
\usepackage{amsthm}
\usepackage{times,amssymb,amsmath,amsfonts,float,nicefrac,color,subfigure,bm}
\usepackage{mathrsfs}
\usepackage{soul} 
\usepackage{color, xcolor} 
\usepackage[colorlinks,
            linkcolor=blue,
            anchorcolor=blue,
            citecolor=blue]{hyperref}
\newtheorem{theorem}{Theorem}

\newtheorem{definition}{Definition}
\newtheorem{remark}{Remark}
\newtheorem{corollary}{Corollary}
\newtheorem{lemma}{Lemma}
\newtheorem{proposition}{Proposition}
\newtheorem{conjecture}{Conjecture}
\soulregister{\cite}7
\soulregister{\ref}7

\journal{Finite Fields and Their Applications}

\begin{document}

\begin{frontmatter}



\title{Deep Holes of Twisted Reed-Solomon Codes\tnoteref{t1}}
\tnotetext[t1]{This paper was presented in part at the 2024 IEEE International Symposium on Information Theory (ISIT)}

\author[1,2,3]{Weijun Fang}

\affiliation[1]{organization={State Key Laboratory of Cryptography and Digital Economy Security, Shandong University},
            city={Qingdao},
            postcode={266237}, 
            country={China}}
            
\affiliation[2]{organization={Key Laboratory of Cryptologic Technology and Information Security, Ministry of Education, Shandong University},
            city={Qingdao},
            postcode={266237}, 
            country={China}}

\affiliation[3]{organization={ School of Cyber Science and Technology, Shandong University},
            city={Qingdao},
            postcode={266237}, 
            country={China}}

\author[4]{Jingke Xu\corref{cor1}}

\affiliation[4]{organization={School of Information Science and Engineering, Shandong Agricultural University},
            city={Tai'an},
            postcode={266237}, 
            country={China}}
\cortext[cor1]{Corresponding author: Jingke Xu (email: xujingke@sdau.edu.cn)}
\author[1,2,3]{Ruiqi Zhu}

\begin{abstract}
The deep holes of a linear code are the vectors that achieve the maximum error distance (covering radius) to the code. {Determining the covering radius and deep holes of linear codes is a fundamental problem in coding theory. In this paper, we investigate the problem of deep holes of twisted Reed-Solomon codes.} The covering radius and a standard class of deep holes of twisted Reed-Solomon codes ${\rm TRS}_k(\mathcal{A}, \theta)$ are obtained for a general evaluation set $\mathcal{A} \subseteq \mathbb{F}_q$. Furthermore, we consider the problem of determining all deep holes of the full-length twisted Reed-Solomon codes ${\rm TRS}_k(\mathbb{F}_q, \theta)$. For 
 even $q$, by utilizing the polynomial method and Gauss sums over finite fields, we prove that the standard deep holes are all the deep holes of ${\rm TRS}_k(\mathbb{F}_q, \theta)$ with $\frac{3q-4}{4} \leq k\leq q-4$. For odd $q$, we adopt a different method and employ the results on some equations over finite fields to show that there are also no other deep holes of ${\rm TRS}_k(\mathbb{F}_q, \theta)$ with $\frac{3q+3\sqrt{q}-7}{4} \leq k\leq q-4$. In addition, for the boundary cases of $k=q-3, q-2$ and $q-1$, we completely determine their deep holes using results on certain character sums.
\end{abstract}




\end{frontmatter}



\section{Introduction}

Let $\mathbb{F}_q$ be the finite field of $q$ elements with characteristic $p$ and $\mathbb{F}^n_q$ be the vector space of dimension $n$ over $\mathbb{F}_q$. The Hamming distance $d(\bm u, \bm v)$ of two vectors $\bm u=(u_0, u_1, \cdots, u_{n-1}), \bm v =(v_0, v_1, \cdots, v_{n-1})\in \mathbb{F}^n_q$ is defined as $d(\bm u, \bm v) \triangleq |\{0 \leq i \leq n-1: u_i \neq v_i \}|.$

An $[n, k, d]_q$-linear code $C$ is just a $k$-dimensional subspace of $\mathbb{F}^n_q$ with minimum distance $d=d(C)$, which is defined as 
$d(C)\triangleq \min\limits_{ \bm u\neq \bm v \in C}d(\bm u, \bm v).$
For any vector $\bm u \in \mathbb{F}^n_q$, we define its error distance to the code $C$ as  $d(\bm u, C)=\min\limits_{\bm c \in C}d(\bm u, \bm c).$
Then the covering radius $\rho(C)$ of $C$ is defined as the maximum error distance, i.e., $\rho(C)=\max\limits_{\bm u \in \mathbb{F}^n_q}d(\bm u, C).$
The vectors achieving maximum error distance are called deep holes of the code. The computation of the covering radius is a fundamental problem in coding theory. However,  McLoughlin\cite{M84} has proven that the computational difficulty of determining the covering radius of random linear codes strictly exceeds NP-completeness. There are few linear codes with known covering radii, and examples include perfect codes \cite{HP03}, generalized Reed-Solomon codes\cite{D94}, and first-order Reed-Muller codes ${\rm RM}(1, m)$ (where $m$ is even)\cite{R76}, among others. Therefore, it is very hard to determine the exact value of the covering radius of a linear code, not to mention deciding its deep holes.
 
In recent years, the problem of determining the deep holes of Reed-Solomon codes has attracted lots of attention in the literature \cite{CM07,LW08,LW08-2,LW10,WH12,ZW12,CMP12,ZFL13,KW16,ZCL16,K17,ZWK20, ZW23}. 
\begin{definition}
    Let $\mathcal{A}=\{\alpha_1, \alpha_2, \cdots, \alpha_n\} \subseteq \mathbb{F}_q$ be the evaluation set, then the Reed-Solomon code ${\rm RS}_k(\mathcal{A})$ of length $n$ and dimension $k$, is defined as 
\[{\rm RS}_k(\mathcal{A})=\{(f(\alpha_1), f(\alpha_2), \cdots, f(\alpha_n)): f(x)\in\mathbb{F}_q[x], \deg(f) \leq k-1\}.\]
\end{definition}

It can be shown that the covering radius of ${\rm RS}_k(\mathcal{A})$ is equal to $n-k$. It was shown in \cite{GV05} that the problem
of determining whether a vector is a deep hole of a
given Reed-Solomon code is NP-hard. 
\begin{definition}
    The generating polynomial of a vector $\bm u \in \mathbb{F}^n_q$ is defined as the unique polynomial $f(x) \in \mathbb{F}_q[x]$ with $\deg(f) \leq n-1$ and $f(\alpha_i)=u_i$ for $i=1,2, \cdots,n$. 
    Conversely, we denote $\bm u_f=(f(\alpha_1), f(\alpha_2), \cdots, f(\alpha_{n})) \in \mathbb{F}^n_q$ for any polynomial $f(x) \in \mathbb{F}_q[x]$ with $\deg(f) \leq n-1$.
\end{definition}  

It can be verified that the vectors with generating polynomial of degree $k$ are deep holes of ${\rm RS}_k(\mathcal{A})$ \cite{CM07}. There may be some other deep holes of ${\rm RS}_k(\mathcal{A})$ for certain subset $\mathcal{A}$ of $\mathbb{F}_q$. However, for the Reed-Solomon codes with full length, i.e., ${\rm RS}_k(\mathbb{F}_q)$, Cheng and Murray \cite{CM07} conjectured that:

\begin{conjecture}[\!\!\cite{CM07}]\label{conj1}
	For $2 \leq k \leq q-2$, all deep holes of ${\rm RS}_k(\mathbb{F}_q)$ have generating polynomials of degree $k$, except when $q$ is even
	and $k= q - 3$.
	\end{conjecture}

Significant research advances have been made regarding Conjecture \ref{conj1}.  Li and Wan \cite{LW08}(resp. Zhang et al. \cite{ZFL13} ) proved that the vectors with generating polynomials of degree $k+1$ (resp. $k+2$) are not deep holes of ${\rm RS}_k(\mathbb{F}_q)$. By employing techniques such as exponential sums, algebraic geometry, and sieving methods, Li and Wan \cite{LW08-2}, Zhu and Wan \cite{ZW12}, Cafure\cite{CMP12} have provided estimates for the error distances of certain vectors, thereby advancing Conjecture \ref{conj1}. When  $2 \leq k \leq p-2$ or $2 \leq q-p \leq k \leq q-3$, Zhuang et al. \cite{ZCL16} proved that Conjecture \ref{conj1} is true.  In particular, Conjecture \ref{conj1} holds for prime fields. Applying Seroussi and Roth’s results on the extension of RS codes \cite{SR86}, Kaipa \cite{K17} proved that Conjecture 1 holds for $k \geq \lfloor \frac{q-1}{2}\rfloor$. In \cite{ZZ23}, the authors provided an excellent survey on the deep hole problem of RS codes.
	
	Since  Beelen et al. first introduced twisted Reed-Solomon (TRS) codes in \cite{BPN17,BPR22}, 
 many coding scholars have studied TRS codes with good properties, including TRS MDS codes, TRS self-dual codes, and TRS LCD codes \cite{HYNL21,LL21,SYLH22,SYLH22-2, ZZT22}.  Beelen et al. \cite{BBP18} proposed the application of TRS codes to the McEliece PKC, suggesting that TRS codes with multiple twists could be used as an alternative to Goppa codes. Lavauzelle et al. \cite{LR20} presented an efficient key-recovery attack on the McEliece-type cryptosystem based on TRS codes. TRS codes have played an increasingly important role in coding theory. As a generalization of Reed-Solomon codes, it is also difficult to determine the covering radius and deep holes of twisted RS codes.  Thus, as a preliminary exploration, we mainly focus a special class of twisted RS codes and obtain some results on their covering radius and deep holes in this paper.

\subsection{Our Results}
Suppose $2 \leq k<n$ and $\mathcal{A}=\{\alpha_1, \alpha_2, \cdots, \alpha_n\} \subseteq  \mathbb{F}_q$.
Let $\theta \neq 0 \in \mathbb{F}_q$ (when $\theta=0$, the TRS code defined below is just the RS code) and 
\begin{equation}\label{S}
    \mathcal{S}_{k,\theta} \triangleq \{f_{k,\theta}(x)\triangleq\sum_{i=0}^{k-2}f_ix^i+f_{k-1}(x^{k-1}+\theta x^k):  f_0, f_1,\cdots, f_{k-1} \in \mathbb{F}_q\}.
\end{equation}
	
\begin{definition}
The TRS code ${\rm TRS}_k(\mathcal{A}, \theta)$ is defined as
\[{\rm TRS}_k(\mathcal{A}, \theta) \triangleq \{(f(\alpha_1), f(\alpha_2), \cdots, f(\alpha_n)): f \in \mathcal{S}_{k,\theta}\}.\]
\end{definition}	

 Firstly, we determine the covering radius and a standard class of deep holes of ${\rm TRS}_k(\mathcal{A}, \theta)$ for a subset $\mathcal{A} \subseteq \mathbb{F}_q$.

 \begin{theorem}\label{thm1}
     Suppose $\mathcal{A} \subseteq \mathbb{F}_q$ with $|\mathcal{A}|=n$ and $1<  k <n$. Then 
     \begin{description}
         \item[(i)] The covering radius $\rho({\rm TRS}_k(\mathcal{A}, \theta))$ of ${\rm TRS}_k(\mathcal{A}, \theta)$ is equal to $n-k$. (See Theorem \ref{thm4})
         \item[(ii)] Suppose $a \in \mathbb{F}^*_q$ , then the vector $\bm u_f$ with generating polynomial $f(x)=ax^{k}+f_{k,\theta}(x)$, where $f_{k,\theta}(x) \in \mathcal{S}_{k,\theta}$ given by Eq. \eqref{S}, is a deep hole of ${\rm TRS}_k(\mathcal{A}, \theta)$. (See Theorem \ref{thm5})
     \end{description}
 \end{theorem}
 
 Furthermore,  we are committed to completely determining the deep holes of the full-length twisted RS code ${\rm TRS}_k(\mathbb{F}_q, \theta)$. For 
$q$ being even or odd, we use different methods respectively to settle its deep holes.

 For even $q$, inspired by the polynomial method outlined in \cite{K17, H19}, we provide a key necessary condition for a vector to be a deep hole of ${\rm TRS}_k(\mathbb{F}_q, \theta)$. Specifically,  a vector $\bm u$ being a deep hole will lead to a corresponding multivariate polynomial vanishing. By the Schwartz-Zippel Lemma,  we can deduce that the polynomial is indeed the zero polynomial when the dimension $k$ satisfies $\frac{3q-4}{4} \leq k \leq q-4$. Combining this with some techniques involving Gauss sums, we obtain the result that the standard deep holes given in Theorem \ref{thm1} are all the deep holes of ${\rm TRS}_k(\mathbb{F}_q, \theta)$ for $\frac{3q-4}{4} \leq k \leq q-4$. For the boundary cases $k=q-3, q-2$, and $q-1$, by using some results on equations over finite fields and character sums, we also completely determine the deep holes of ${\rm TRS}_k(\mathbb{F}_q, \theta)$. We summary the results for even $q$ as follows.

 \begin{theorem}\label{even}
Suppose $q$ is even with $q \geq 8$, $\bm u \in \mathbb{F}^q_q$ and ${\rm Tr}(\cdot)$ is the absolute trace function from $\mathbb{F}_q$ to $\mathbb{F}_2$.
\begin{description}
    \item[(i)] If $\frac{3q-4}{4}\leq k \leq q-4$ or $k=q-1$, then $\bm u$ is a deep hole of ${\rm TRS}_k(\mathbb{F}_q, \theta)$ if and only if $\bm u$ is generated by Theorem \ref{thm1}. (See Theorems \ref{even} and \ref{thm7})
    \item[(ii)] If $k=q-2$, then $\bm u$ is a deep hole of ${\rm TRS}_k(\mathbb{F}_q, \theta)$ if and only if $\bm u$ is generated by  $w_0x^{q-1}+w_1x^{q-2}+f_{q-2,\theta}(x)$ with $w_0=0, w_1 \neq 0$ or $w_0 \neq 0$ and ${\rm Tr}(\frac{w_1\theta}{w_0})=1$, where $f_{q-2,\theta}(x) \in \mathcal{S}_{q-2, \theta}$. (See Theorem \ref{thm7})
    \item[(iii)] If $k=q-3$, then $\bm u$ is a deep hole of ${\rm TRS}_k(\mathbb{F}_q, \theta)$ if and only if $\bm u$ is given by  Theorem \ref{thm1} or generated by $a(x^{q-2}+\theta^{-1}x^{q-3})+f_{q-3,\theta}(x)$ with $a \neq 0$, $f_{q-3,\theta}(x) \in \mathcal{S}_{q-3, \theta}$ and $2 \nmid m$. (See Theorem \ref{thm7})
\end{description}
\end{theorem}

For odd $q$, the method used in the even case does not work. Instead, we adopt a different method based on Seroussi and Roth's results for the MDS extensions of RS codes. By empolying certain results concerning equations over finite fields, we prove that the standard deep holes given in Theorem \ref{thm1} are all the deep holes for $\frac{3q+3\sqrt{q}-7}{4} \leq k \leq q-4$. The deep holes of ${\rm TRS}_k(\mathbb{F}_q, \theta)$ are also completely determined for the boundary cases $k=q-3$, $q-2$ and $q-1$. Specifically, 

\begin{theorem}\label{odd}
Let $q$ be an odd prime power with $q>16$. Suppose $\bm u \in \mathbb{F}^q_q$ and $\eta$ is the quadratic character of $\mathbb{F}_q$.
\begin{description}
    \item[(i)] If $\frac{3q+3\sqrt{q}-7}{4} \leq k \leq q-4$ or $k=q-1$, then $\bm u$ is a deep hole of ${\rm TRS}_k(\mathbb{F}_q, \theta)$ if and only if $\bm u$ is generated by Theorem \ref{thm1}. (See Theorems \ref{thm9} and \ref{thm10})
    \item[(ii)] If $k=q-2$, then $\bm u$ is a deep hole of ${\rm TRS}_k(\mathbb{F}_q, \theta)$ if and only if $\bm u$ is generated by  $w_0x^{q-1}+w_1x^{q-2}+f_{q-2,\theta}(x)$ with $w_0=0, w_1 \neq 0$ or $w_0 \neq 0$, $f_{q-2,\theta}(x) \in \mathcal{S}_{q-2, \theta}$ and $\eta(w_0^2-4w_0w_1\theta)=-1$. (See Theorem \ref{thm10})
    \item[(iii)] If $k=q-3$, then $\bm u$ is a deep hole of ${\rm TRS}_k(\mathbb{F}_q, \theta)$ if and only if $\bm u$ is given by  Theorem \ref{thm1} or generated by $a(x^{q-2}+\frac{1}{3\theta}x^{q-3})+f_{q-3,\theta}(x)$ with $a \neq 0$, $f_{q-3,\theta}(x) \in \mathcal{S}_{q-3, \theta}$ and $\eta(-3)=-1$. (See Theorem \ref{thm10})
\end{description}
\end{theorem}

The rest of this paper is organized as follows. In Section \ref{sec2}, we present some results on some determinants and character sums. In Section \ref{sec3}, we determine the covering radius and standard deep holes of twisted RS codes.
In Section \ref{sec4}, we present the results on the completeness of deep holes of ${\rm TRS}_k(\mathbb{F}_q, \theta)$. Section \ref{sec5} concludes the paper and provides a conjecture on the deep holes of ${\rm TRS}_k(\mathbb{F}_q, \theta)$.

\section{Preliminaries}\label{sec2}
We first list some fixed notations that will be used throughout this paper.

\begin{itemize}
    \item ${\rm Char}(\mathbb{F}_{q})$ is the characteristic of $\mathbb{F}_{q}$, and $\mathbb{F}_{q}^{*}=\mathbb{F}_{q}\setminus\{0\}$. 
    \item For  $\alpha \in \mathbb{F}_q$,  $\bm c_{n, \theta}(\alpha)=(1, \alpha, \cdots,\alpha^{n-2}, \alpha^{n-1}-\theta \alpha^{n})^\top \in \mathbb{F}^n_q$;  $\bm c_n(\alpha)=(1, \alpha, \cdots, \alpha^{n-1})^{\top} \in \mathbb{F}^n_q$; $\bm c_n(\infty)=(0, \cdots,0,1)^{\top} \in \mathbb{F}^n_q$.
    \item $V(\alpha_1, \alpha_2, \cdots, \alpha_{n})=\prod\limits_{1\leq i <j \leq n}(\alpha_j-\alpha_i)$, the determinant of the Vandermonde matrix defined by $\alpha_1, \alpha_2, \cdots, \alpha_{n}$.
    \item $\det({\bm v}_1|{\bm v}_2|\cdots|{\bm v}_n)$, the determinant of the matrix defined by $n$ column vectors ${\bm v}_1,{\bm v}_2,\cdots,{\bm v}_n $ in $\mathbb{F}_q^n$.
    \item For $f\in\mathbb{F}_q[X_1,X_2,\cdots,X_n]$, $${\rm N}(f)=|\{(X_1,\cdots,X_n)\in\mathbb{F}_q^n:f(X_1,\cdots,X_n)=0\}|.$$
\end{itemize}

\subsection{Some Determinants}
In this subsection, we present the calculations of some determinants, which will be used in our proof of main theorems.

 \begin{lemma}[{\cite[Sec. 338]{M60}}]\label{det2}
 For any $\alpha_1, \alpha_2, \cdots, \alpha_n \in \mathbb{F}_q$, we have 
  \[\det\left(\begin{array}{cccc}
		1 & 1 &\cdots & 1 \\
		\alpha_1 & \alpha_2 &\cdots & \alpha_{n} \\
		\vdots & \vdots &\ddots & \vdots \\
		\alpha_1^{n-2} & \alpha_2^{n-2} &\cdots & \alpha_{n}^{n-2} \\
		\alpha_1^{n} & \alpha_2^{n} &\cdots & \alpha_{n}^{n} \\  
	\end{array}\right)=V(\alpha_1, \alpha_2, \cdots, \alpha_{n})\sum_{i=1}^n\alpha_i,\]
 and
    \[\det\left(\begin{array}{cccc}
		1 & 1 &\cdots & 1 \\
		\alpha_1 & \alpha_2 &\cdots & \alpha_{n} \\
		\vdots & \vdots &\ddots & \vdots \\
		\alpha_1^{n-2} & \alpha_2^{n-2} &\cdots & \alpha_{n}^{n-2} \\
		\alpha_1^{n+1} & \alpha_2^{n+1} &\cdots & \alpha_{n}^{n+1} \\   
	\end{array}\right)=V(\alpha_1, \alpha_2, \cdots, \alpha_{n})\sum_{1 \leq i \leq j  \leq n}\alpha_i\alpha_j.\]
\end{lemma}

 \begin{lemma}\label{det3}
 For any $\alpha_1, \alpha_2, \cdots, \alpha_n \in \mathbb{F}_q$,
\[\det\big(\bm c_{n, \theta}(\alpha_1)|\bm c_{n, \theta}(\alpha_2)|\cdots|\bm c_{n, \theta}(\alpha_{n})\big)\\	
=V(\alpha_1, \alpha_2, \cdots, \alpha_{n})\big(1-\theta\sum_{i=1}^{n}\alpha_i\big).\]
\end{lemma}

\begin{proof}
By Lemma \ref{det2}, we have
\begin{eqnarray*}
    &&\det\big(\bm c_{n, \theta}(\alpha_1)|\bm c_{n, \theta}(\alpha_2)|\cdots|\bm c_{n, \theta}(\alpha_{n})\big)\\
    &=&V(\alpha_1, \alpha_2, \cdots, \alpha_{n})-\theta\det\left(\begin{array}{ccc}
		1 & \cdots & 1 \\
		\alpha_1 & \cdots & \alpha_{n} \\
		\vdots & \ddots & \vdots \\
		\alpha_1^{n-2} & \cdots & \alpha_{n}^{n-2} \\
		\alpha_1^{n} & \cdots & \alpha_{n}^{n} \\  
	\end{array}\right)\\
 &=& V(\alpha_1, \alpha_2, \cdots, \alpha_{n})\big(1-\theta\sum_{i=1}^{n}\alpha_i\big).
\end{eqnarray*}
\end{proof}

\subsection{Character Sums}
In this subsection, we present some basic notations and results of group characters and exponential sums \cite{LN97}. Suppose $p={\rm Char}(\mathbb{F}_q)$ and $q=p^m$, the absolute trace function ${\rm Tr}(x)$  is defined by 
$$
{\rm Tr}(x)=x+x^p+x^{p^2}+\cdots+x^{p^{m-1}} .
$$

An additive character $\chi$ of $\mathbb{F}_q$ is a homomorphism from $\mathbb{F}_q$ into the multiplicative group $U$ of complex numbers of absolute value 1. For any $a \in \mathbb{F}_q$, the function
$$
\chi_a(x)=\zeta_p^{{\rm Tr}(a x)},
$$
defines an additive character of $\mathbb{F}_q$, where $\zeta_p=e^{\frac{2 \pi \sqrt{-1}}{p}}$. For $a=0$, $\chi_a(x)$ is called the trivial additive character of $\mathbb{F}_q$, for which  $\chi_0(x)=1$ for all $x \in \mathbb{F}_q$. For $a=1$,  $\chi_1(x)=\zeta_p^{{\rm Tr}(x)}$ is called the canonical additive character of $\mathbb{F}_q$. 

Characters of the multiplicative group $\mathbb{F}_q^*$ of $\mathbb{F}_q$ are called multiplicative characters of $\mathbb{F}_q$. It is known that all characters of $\mathbb{F}_q^*$ are given by
$$
\psi_i\left(\pi^j\right)=\zeta_{q-1}^{i j} \text { for } j=0,1, \ldots, q-2,
$$
where $0 \leq i \leq q-2$ and $\pi$ is a primitive element of $\mathbb{F}_q$. It is convenient to extend the definition of $\psi_i$ by setting $\psi_i(0)=0$. For $i=0$, $\psi_0(x)=1$ for all $x \in \mathbb{F}^*_q$ is called the trivial multiplicative character of $\mathbb{F}_q$.  For $i=(q-1) / 2$, the multiplicative character $\psi_{(q-1) / 2}$ is called the quadratic character of $\mathbb{F}_q$, and is denoted by $\eta$ in this paper. That is $\eta(x)=1$ if $x \in \mathbb{F}_q^*$ is a square; $\eta(x)=-1$ if $x \in \mathbb{F}_q^*$ is not a square. 

Let $\psi$ be a multiplicative character and $\chi$ an additive character of $\mathbb{F}_q$. The $\operatorname{Gauss} \operatorname{sum} G(\psi, \chi)$ is defined by
$$
G(\psi, \chi)=\sum_{x \in \mathbb{F}_q^*} \psi(x) \chi(x).
$$

\begin{proposition}[\cite{LN97}]\label{gauss}
\begin{description}
    \item[(i)]$G(\psi, \chi_{ab})=\overline{\psi(a)}G(\psi, \chi_b)$, for any $a \in \mathbb{F}_q^*$ and $b \in \mathbb{F}_q$.
    \item[(ii)] If $\psi \neq \psi_0$ and $\chi \neq \chi_0$, then $|G(\psi, \chi)|=\sqrt{q}$.
\end{description}
    
\end{proposition}

 \begin{proposition}[{\cite[Theorems 5.32, 5.33 and 5.34]{LN97}}]\label{prop2}
     Let $\chi\neq \chi_0$ be a nontrivial additive character of $\mathbb{F}_q$.
     \begin{description}
         \item[(i)] Suppose $n \in \mathbb{N}$, and $d=\operatorname{gcd}(n, q-1)$. Then
         $$\left|\sum_{c \in \mathbb{F}_q} \chi\left(a c^n+b\right)\right| \leqslant(d-1) q^{1 / 2}$$
for any $a, b \in \mathbb{F}_q$ with $a \neq 0$.
\item[(ii)] Suppose $q$ odd, let $f(x)=a_2 x^2+a_1 x+a_0 \in \mathbb{F}_q[x]$ with $a_2 \neq 0$. Then
$$
\sum_{c \in \mathbb{F}_q} \chi(f(c))=\chi\left(a_0-a_1^2\left(4 a_2\right)^{-1}\right) \eta\left(a_2\right) G(\eta, \chi),
$$
where $\eta$ is the quadratic character of $\mathbb{F}_q$.
\item[(iii)] Suppose $q$ is even and $b\in \mathbb{F}_q^*$.  Let $f(x)=a_2 x^2+a_1 x+a_0 \in \mathbb{F}_q[x]$ with $a_2 \neq 0$. Then
$$
\sum_{c \in \mathbb{F}_q} \chi_b(f(c))=\begin{cases}\chi_b(a_0)q & {\rm if~} ba_2+b^2a_1^2=0,\\
0 &{\rm otherwise}.
\end{cases}
$$
\end{description}
     
 \end{proposition}

\begin{proposition}[{\cite[Theorem 5.41]{LN97}}]\label{prop3}
Let $\psi$ be a multiplicative character of $\mathbb{F}_q$ of order $m>1$ and let $f \in \mathbb{F}_q[x]$ be a monic polynomial of positive degree that is not an $m$-th power of a polynomial. Let $d$ be the number of distinct roots of $f$ in its splitting field over $\mathbb{F}_q$. Then for every $a \in \mathbb{F}^*_q$, we have
$$
\left|\sum_{c \in \mathbb{F}_q} \psi(a f(c))\right| \leqslant(d-1) q^{1 / 2}.
$$	
\end{proposition}
     
\begin{proposition}[{\cite[Lemma 6.24]{LN97}} ]\label{prop4}
     For odd $q$, let $b\in \mathbb{F}_q,a_1,a_2\in\mathbb{F}^*_q$, and $\eta$ be the quadratic character of $\mathbb{F}_q$. Then, ${\rm N}(a_1X^2+a_2Y^2-b)=q+v(b)\eta(-a_1a_2),$
     where $v(0)=q-1$ and $v(b)=-1$ for $b\in \mathbb{F}^*_q$.
\end{proposition}

\section{Covering Radius and Standard Deep Holes of ${\rm TRS}_k(\mathcal{A}, \theta)$}\label{sec3}
In this section, we determine the covering radius and a class of deep holes of ${\rm TRS}_k(\mathcal{A}, \theta)$ by using the redundancy bound and generator matrix of the code.

Suppose $\mathcal{A} \subseteq \mathbb{F}_q$ with $|\mathcal{A}|=n$ and $1<  k <n$. It is obvious that the dimension of ${\rm TRS}_k(\mathcal{A}, \theta)$ is $k$. Note that ${\rm TRS}_k(\mathcal{A}, \theta)$ is contained in the Reed-Solomon code ${\rm RS}_{k+1}(\mathcal{A})$, thus the minimum distance 
\[d({\rm TRS}_k(\mathcal{A}, \theta))\geq d({\rm RS}_{k+1}(\mathcal{A}))=n-k.\]
From the definition,  ${\rm TRS}_k(\mathcal{A}, \theta)$ has a generator matrix
\[G=\left(\begin{array}{cccc}
1 & 1 &\cdots & 1 \\
\alpha_1 & \alpha_2 &\cdots & \alpha_n \\
\vdots & \vdots &\ddots & \vdots \\
\alpha_1^{k-2} & \alpha_2^{k-2} &\cdots & \alpha_n^{k-2} \\
\alpha_1^{k-1}+\theta \alpha_1^k & \alpha_2^{k-1}+\theta \alpha_2^k &\cdots & \alpha_n^{k-1}+\theta \alpha_n^k \\
\end{array}\right).\]

The covering radius of ${\rm TRS}_k(\mathcal{A}, \theta)$ then can be easily determined.
\begin{theorem}\label{thm4}
  The covering radius $\rho({\rm TRS}_k(\mathcal{A}, \theta))$ of ${\rm TRS}_k(\mathcal{A}, \theta)$ is equal to $n-k$.
\end{theorem}
	
\begin{proof}
Since $\dim_{\mathbb{F}_q}({\rm TRS}_k(\mathcal{A}, \theta))=k$, by the redundancy bound \cite[Corollary 11.1.3]{HP03}, we have $\rho({\rm TRS}_k(\mathcal{A}, \theta)) \leq n-k$. Note that ${\rm TRS}_k(\mathcal{A}, \theta)$ is a subcode of the Reed-Solomon code ${\rm RS}_{k+1}(\mathcal{A})$, then by the supercode lemma \cite[Lemma 11.1.5]{HP03}, we have $\rho({\rm TRS}_k(\mathcal{A}, \theta)) \geq d({\rm RS}_{k+1}(\mathcal{A})) =n-k$. The conclusion then follows.
\end{proof}
	
An equivalent condition under which a vector is a deep hole of ${\rm TRS}_k(\mathcal{A}, \theta)$ is given as follows.
	
\begin{proposition}\label{prop5}
Suppose $G$ is a generator matrix of ${\rm TRS}_k(\mathcal{A}, \theta)$ and $\bm u \in \mathbb{F}^n_q$, then $\bm u$ is a deep hole of ${\rm TRS}_k(\mathcal{A}, \theta)$ if and only if the matrix $G'=\left(\begin{array}{c}
			G\\
			\bm u
		\end{array}\right)$ generates an $[n, k+1]$-MDS code  over $\mathbb{F}_q$.
\end{proposition}
	
\begin{proof} For convenience, denote $C={\rm TRS}_k(\mathcal{A}, \theta)$ and $C'$ as the  code generated by $G'$.
Then, $d(C')=\min\{d(C),d(\bm u,C)\}$.
By Theorem \ref{thm4}, it has $d(C)\geq n-k=\rho(C)\geq d(\bm u,C)$, hence $d(C')= d(\bm u,C)$.
So, $C'$ is an $[n,k+1]$-MDS code over $\mathbb{F}_q$ if and only if $d(\bm u,C)=d(C')=n-k$, which is equivalent to ${\bm u}$ is a deep hole of $C$. 
\end{proof}
	
By Proposition \ref{prop5}, we can obtain a class of deep holes of ${\rm TRS}_k(\mathcal{A}, \theta)$, which we call standard deep holes.
	
\begin{theorem}[Standard Deep Holes]\label{thm5}
Suppose $a \in \mathbb{F}^*_q$ , then the vector $\bm u_f$ with generating polynomial $f(x)=ax^{k}+f_{k,\theta}(x)$, where $f_{k,\theta}(x) \in \mathcal{S}_{k,\theta}$ given by Eq. \eqref{S}, is a deep hole of ${\rm TRS}_k(\mathcal{A}, \theta)$.
\end{theorem}
	
\begin{proof}
By the definition of ${\rm TRS}_k(\mathcal{A}, \theta)$,  the vector with generating polynomial $f_{k,\theta}(x)$ is a codeword of  ${\rm TRS}_k(\mathcal{A}, \theta)$. Thus we only need to prove the theorem for the case of $f(x)=ax^{k}$.
At this moment, it is easy to see that  $\left(\begin{array}{c}
			G\\
			\bm u_f
		\end{array}\right)$  is row equivalent to 
\[\left(\begin{array}{cccc}
			1  &1  &\cdots & 1 \\
			\alpha_1  &\alpha_2  &\cdots & \alpha_n \\
			\vdots & \vdots &\ddots & \vdots \\
			
			\alpha_1^{k-1}  &\alpha_2^{k-1}  &\cdots & \alpha_n^{k-1} \\
			\alpha_1^k  &\alpha_2^k  &\cdots &  \alpha_n^k \\
	\end{array}\right),\]
which generates the RS code ${\rm RS}_{k+1}(\mathcal{A})$. The conclusion then follows from Proposition \ref{prop5}.

\end{proof}

\section{On the completeness of Deep Holes of ${\rm TRS}_k(\mathbb{F}_{q}, \theta)$}\label{sec4}
In this section, we devote to presenting our main theorems on the completeness of deep holes of ${\rm TRS}k(\mathbb{F}{q}, \theta)$.
The results will be divided into two cases: when $q$ is even and when $q$ is odd. Before that, we demonstrate a criterion for a vector to be a deep hole of a linear code.

\subsection{The relation between deep holes and their syndromes} 
 Firstly, we give some results on the deep holes of a linear code and its syndrome.  If $\bm u$ is a deep hole of a linear code $C$, then for any $a \in \mathbb{F}^*_q$ and $\bm c \in C$,  both of $a \bm u$ and $\bm u+ \bm c$ are deep holes. Suppose $H$ is a parity-check matrix of $C$. Then $H\cdot(a \bm u)^\top=aH\cdot \bm u^\top$ and $H\cdot(\bm u+ \bm c)=H\cdot \bm u^\top$. 
	
\begin{proposition}[{\cite[Theorem 11.1.2]{HP03}}]\label{prop6}
Let $C$ be an $[n,k]$-linear code with a parity-check matrix $H$ and covering radius $\rho(C)$. Suppose $\bm u \in  \mathbb{F}^n_q$, then $\bm u$ is a deep hole of $C$ if and only if $H\cdot \bm u^\top$ can not be expressed as a linear combination of any $\rho(C)-1$ columns of $H$ over  $\mathbb{F}_q$.
\end{proposition}

Throughout this paper, we denote $r=q-k$. Suppose $\mathbb{F}_q=\{\alpha_1, \alpha_2, \cdots, \alpha_q \}$. It is easy to verify that ${\rm TRS}_k(\mathbb{F}_q, \theta)$ has a parity-check matrix 
\begin{equation}\label{pc}
		H=\left(\begin{array}{ccc}
			1 & \cdots & 1 \\
			\alpha_1 & \cdots & \alpha_q \\
			\vdots & \ddots & \vdots \\
			\alpha_1^{r-2} & \cdots & \alpha_q^{r-2} \\
			\alpha_1^{r-1}-\theta \alpha_1^{r} & \cdots & \alpha_q^{r-1}-\theta\alpha_q^{r} \\
		\end{array}\right).
\end{equation}
	
The following proposition provides the relationship between the generating polynomial of a vector and its syndrome.

\begin{proposition}\label{prop7}
Suppose ${\bm w}=(w_0,w_1,\cdots,w_{r-1})\in\mathbb{F}_q^r$  and $H$ is given by Eq. (\ref{pc}). Then, the solutions of the linear system $H \cdot \bm u^\top=\bm w^\top$  are the vectors $\bm u_f$ generated by $f(x)=-\sum_{i=0}^{r-1}w_ix^{q-1-i}+f_{k,\theta}(x)$, where $f_{k,\theta}(x)\in \mathcal{S}_{k,\theta}$.   
\end{proposition}

\begin{proof}
By the definition of TRS codes, it is sufficient to show that $H\bm u^\top_f=\bm w^\top$ for the case of $f(x)=-\sum_{i=0}^{r-1}w_ix^{q-1-i}$. The conclusion then follows from the fact that for all positive integers $j$, we have 
$$\sum^q_{i=1}\alpha_i^j=\begin{cases}
    -1, \text{ if } q-1 \mid j \text{ and } j\neq 0;\\
0, ~~~\text{ otherwise }.\end{cases}$$
\end{proof}

From Theorem \ref{thm5} and Proposition \ref{prop7}, we can obtain the following result, which characterizes the syndromes of standard deep holes of ${\rm TRS}_k(\mathbb{F}_q, \theta)$.
\begin{corollary}\label{cor1}
    Suppose $\bm u \in \mathbb{F}^q_q$. Then, $\bm u$ is the standard deep hole of ${\rm TRS}_k(\mathbb{F}_q, \theta)$ if and only if its syndrome $\bm w^\top=H \cdot \bm u^\top$ is of the form $(0, 0, \cdots, 0, w_{r-1})^\top$ with $w_{r-1} \neq 0$.
\end{corollary}

  In the following, we provide an equivalent characterization of the deep hole of ${\rm TRS}_k(\mathbb{F}_q, \theta)$ from the perspective of its syndrome, which will be useful in our main theorems.
\begin{proposition}\label{prop8}
Suppose $\bm u \in \mathbb{F}^q_q$, then $\bm u$ is a deep hole of ${\rm TRS}_k(\mathbb{F}_q, \theta)$ if and only if $\bm c_{r, \theta}(\alpha_1), \bm c_{r, \theta}(\alpha_2), \cdots, \bm c_{r, \theta}(\alpha_{r-1})$ and $H\cdot \bm u^\top$ are $\mathbb{F}_q$-linearly independent for any distinct elements $\alpha_1, \alpha_2, \cdots, \alpha_{r-1} \in \mathbb{F}_q$. 
\end{proposition}
	
\begin{proof}
Note that the minimum distance $d({\rm TRS}_k(\mathbb{F}_q)) \geq q-k=\rho({\rm TRS}_k(\mathbb{F}_q))$. Thus any $\rho({\rm TRS}_k(\mathbb{F}_q))-1=r-1$ columns of $H$ are linearly independent. Then the conclusion immediately follows from  Proposition \ref{prop6}.
\end{proof}

\subsection{Deep Holes of ${\rm TRS}_k(\mathbb{F}_q, \theta)$ for Even $q$}\label{sec4.2}
Inspired by the polynomial method proposed in \cite{H19}, in this subsection, for $q$ is even, we determine all deep holes of ${\rm TRS}_k(\mathbb{F}_q, \theta)$ when $ \frac{3q-4}{4} \leq k \leq q-1$.

For any $x_1, x_2, \cdots, x_{n} \in \mathbb{F}_q$, denote 
\[\prod\limits_{i=1}^{n}(X-x_i)=\sum\limits_{i=0}^{n}S_{i}(x_1, x_2, \cdots, x_{n})X^i,\]
where $S_{i}(x_1, \cdots, x_{n})=(-1)^{n-i}\sum\limits_{1\leq j_1<j_2<\cdots<j_{n-i}\leq n}\prod\limits_{\ell=1}^{n-i}x_{j_\ell}$ for all $0\leq i\leq n$.

Let $\bm w=(w_0, w_1, \cdots, w_{r-1}) \in  \mathbb{F}^r_q$. Denote 
\[f=f(x_1, \cdots, x_{r-2})=\sum\limits_{i=0}^{r-2}w_i S_{i}(x_1, \cdots, x_{r-2}), \]
	\[g=g(x_1,  \cdots, x_{r-2})=\sum\limits_{i=1}^{r-2}w_{i}S_{i-1}(x_1, \cdots, x_{r-2}),\]
	and 
 \[h=h(x_1,\dots,x_{r-2})=\sum_{i=2}^{r-2}w_iS_{i-2}(x_1,  \cdots, x_{r-2})+\theta^{-1}w_{r-1}.\]
 
From Proposition \ref{prop8}, $\bm u$ is a deep hole of ${\rm TRS}_k(\mathbb{F}_q, \theta)$ if and only if 
\[\det\big(\bm c_{r, \theta}(\alpha_1)|\bm c_{r, \theta}(\alpha_2)|\cdots|\bm c_{r, \theta}(\alpha_{r-1})|\bm w^\top\big) \neq 0\] 
for any pairwise distinct elements $\alpha_1, \alpha_2, \cdots, \alpha_{r-1} \in \mathbb{F}_q$, where $\bm w^{\top}= H \cdot \bm u^{\top}$. We will regard these elements $\alpha_1, \alpha_2, \cdots, \alpha_{r-1}$ as variables $x_1, x_2, \cdots, x_{r-1}$, and give a formula of $\det\big(\bm c_{r, \theta}(x_1)|\bm c_{r, \theta}(x_2)|\cdots|\bm c_{r, \theta}(x_{r-1})|\bm w^\top\big)$ as follows.
\begin{lemma}\label{lem4}
With the notations above, we have
\begin{eqnarray*}
   && \frac{\det\big(\bm c_{r, \theta}(x_1)|\bm c_{r, \theta}(x_2)|\cdots|\bm c_{r, \theta}(x_{r-1})|\bm w^\top\big)}{\theta V(x_1, x_2, \cdots, x_{r-1})}\\
   &=&f\cdot x^2_{r-1}+\big(\theta^{-1}+\sum_{i=1}^{r-2}x_i\big)f\cdot x_{r-1}+\big(\theta^{-1}+\sum_{i=1}^{r-2}x_i\big)g+ h.
\end{eqnarray*}

\end{lemma}

\begin{proof}
By Lemma \ref{det3} and ${\rm Char}(\mathbb{F}_{q})=2$, we have
\begin{eqnarray*}
    &&\det\big(\bm c_{r, \theta}(x_1)|\bm c_{r, \theta}(x_2)|\cdots|\bm c_{r, \theta}(x_{r-1})|\bm c_{r, \theta}(X)\big)\\
    &=&V(x_1, x_2, \cdots, x_{r-1})\prod_{j=1}^{r-1}(X+x_j)\big(1+\theta(X+\sum_{i=1}^{r-1}x_i)\big).
\end{eqnarray*}
For convenience, denote $S_{i, n}=S_i(x_1, \cdots, x_{n})$.	Then we have
\begin{eqnarray*}
    &&\det\big(\bm c_{r, \theta}(x_1)|\bm c_{r, \theta}(x_2),|\cdots|\bm c_{r, \theta}(x_{r-1})|\bm w^\top\big)\\
    &=&\theta V(x_1, x_2, \cdots, x_{r-1})\big(\sum_{i=0}^{r-2}w_iD_i+\theta^{-1}w_{r-1}\big),
\end{eqnarray*}
where $D_i$ is the coefficient of $X^i$ in $\prod\limits_{j=1}^{r-1}(X+x_j)\big(X+\sum\limits_{t=1}^{r-1}x_t+\theta^{-1}\big)$, i.e., $D_i=S_{i, r-1}\big(\theta^{-1}+\sum\limits_{t=1}^{r-1}x_t\big)+ S_{i-1,r-1}.$
		Note that $S_{\ell, r-1}=x_{r-1}S_{\ell, r-2}+S_{\ell-1, r-2}.$
		Thus 
		\begin{eqnarray*}
			D_i&=&(x_{r-1}S_{i,r-2}+S_{i-1,r-2})\big(x_{r-1}+\sum_{t=1}^{r-2}x_t +\theta^{-1}\big)+ (x_{r-1}S_{i-1,r-2}+S_{i-2,r-2})\\
			&=&S_{i,r-2}x^2_{r-1}+(\theta^{-1}+\sum_{t=1}^{r-2}x_t)S_{i,r-2}x_{r-1}+S_{i-1,r-2}(\theta^{-1}+\sum_{t=1}^{r-2}x_t)+ S_{i-2,r-2}.
		\end{eqnarray*}
		Therefore,
		\begin{eqnarray*}
			&&\frac{\det\big(\bm c_{r, \theta}(x_1)|\bm c_{r, \theta}(x_2)|\cdots|\bm c_{r, \theta}(x_{r-1})|\bm w^\top\big)}{\theta V(x_1, x_2, \cdots, x_{r-1})}\\
			&=&\sum_{i=0}^{r-2}w_i\Big(S_{i,r-2}x^2_{r-1}+(\theta^{-1}+\sum_{t=1}^{r-2}x_t)S_{i,r-2}x_{r-1} +S_{i-1,r-2}(\theta^{-1}+\sum_{t=1}^{r-2}x_t)\\
   &&+ S_{i-2,r-2}\Big)+\theta^{-1}w_{r-1}\\
			&=&fx^2_{r-1}+(\theta^{-1}+\sum_{i=1}^{r-2}x_i)(f x_{r-1} +g)+ h.
		\end{eqnarray*}\end{proof}
From Lemma \ref{lem4} and Proposition \ref{prop8}, we provide a necessary condition under which $\bm u$ is a deep hole of ${\rm TRS}_k(\mathbb{F}_q, \theta)$. 

	Denote $\tilde{f}=f(x_1, \cdots,x_{r-3},\sum\limits_{j=1}^{r-3}x_j+\theta^{-1}), \tilde{h}= h(x_1, \cdots,x_{r-3}, \sum\limits_{j=1}^{r-3}x_j+\theta^{-1})$.
\begin{lemma}\label{lem5}
Suppose $r-3 \geq 1$, i.e., $k \leq q-4$. Let $\bm u \in \mathbb{F}^n_q$ and $\bm w^{\top}=(w_0, w_1, \cdots, w_{r-1})^{\top}=H\cdot \bm u^{\top} \in  \mathbb{F}^r_q$. If $\bm u$ is a deep hole of ${\rm TRS}_k(\mathbb{F}_q, \theta)$, then the polynomial 
\begin{eqnarray*}
    P(x_1, \cdots, x_{r-3})&=&V(x_1, \cdots, x_{r-3})\cdot\prod_{t=1}^{r-3}\big(\sum_{j=1}^{r-3}x_j+\theta^{-1}+x_t\big)\cdot\prod_{i=1}^{r-3}\big(\tilde{f}x_i^2+\tilde{h}\big)\\
    &&\cdot\tilde{f}\cdot\big(\tilde{f}(\sum_{j=1}^{r-3}x_j+\theta^{-1})^2+\tilde{h}\big) \in \mathbb{F}_q[x_1, x_2, \cdots, x_{r-3}]
\end{eqnarray*}
vanishes on $\underbrace{\mathbb{F}_q\times \mathbb{F}_q \times \cdots \times \mathbb{F}_q}_{r-3}.$
\end{lemma}
	
\begin{proof}
For any $x_1, \cdots, x_{r-3} \in \mathbb{F}_q$, if $x_i=x_j$ for some $1 \leq i < j \leq r-3$, then $V(x_1, x_2, \cdots, x_{r-3})=0$. If $\sum\limits_{j=1}^{r-3}x_j+\theta^{-1}+x_t=0$ for some $1 \leq t \leq r-3$, then $\prod\limits_{t=1}^{r-3}(\sum\limits_{j=1}^{r-3}x_j+\theta^{-1}+x_t)=0$. So we set
$x_{r-2}=\sum\limits_{j=1}^{r-3}x_j+\theta^{-1},$
and suppose $x_1, x_2, \cdots, x_{r-2}$ are pairwise distinct. If $\tilde{f}=0$, we are done. So we assume that $\tilde{f}\neq 0$.  
Since $\bm u$ is a deep hole of ${\rm TRS}_k(\mathbb{F}_q, \theta)$, by Proposition \ref{prop8}, we have $\det\big(\bm c_{r, \theta}(x_1)|\bm c_{r, \theta}(x_2)|\cdots|\bm c_{r, \theta}(x_{r-1})|\bm w^\top\big) \neq 0$ for any distinct elements $x_1, x_2, \cdots, x_{r-1} \in \mathbb{F}_q$. From $\sum\limits_{j=1}^{r-2}x_j+\theta^{-1} =0$ and Lemma \ref{lem4}, we obtain that
\[\frac{\det\big(\bm c_{r, \theta}(x_1)|\cdots|\bm c_{r, \theta}(x_{r-1})|\bm w^\top\big)}{\theta V(x_1, x_2, \cdots, x_{r-1})}=\tilde{f}x^2_{r-1}+\tilde{h}.\] 
Thus $\tilde{f}x^2_{r-1}+\tilde{h} \neq 0$ for any  $x_{r-1} \in \mathbb{F}_q\backslash \{x_1, x_2, \cdots, x_{r-2}\}$. Since $\mathbb{F}_q$ has characteristic 2, the equation $\tilde{f}X^2+\tilde{h} = 0$ has a unique solution $X \in \mathbb{F}_q$. Thus we can deduce that the solution can only be one of $x_1, \cdots, x_{r-2}$, i.e., $\tilde{f}x_i^2+\tilde{h} = 0$ for some $1 \leq i \leq r-2$, that is, 
\[\prod_{i=1}^{r-3}(\tilde{f} x_i^2+\tilde{h})\big(\tilde{f}(\sum_{j=1}^{r-3}x_j+\theta^{-1})^2+\tilde{h}\big)=0.\]\end{proof}

\begin{remark} We explain here the difference between our approach and the method proposed in \cite{H19}, which considers the deep hole problem of RS codes.  In \cite{H19}, the authors presented an expression for the term $\frac{\det\big(\bm c_{r-1}(x_1)|\bm c_{r-1}(x_2)|\cdots|\bm c_{r-1}(x_{r-1})|\bm w^\top\big)}{V(x_1, x_2, \cdots, x_{r-1})}$ as a linear polynomial in $x_{r-1}$ for given $x_1, \cdots, x_{r-2}$. The key to their method is the fact that a linear equation always has a solution over $\mathbb{F}_q$. However, this approach does not extend to our setting. In our case, the corresponding term, $\frac{\det\big(\bm c_{r, \theta}(x_1)|\bm c_{r, \theta}(x_2)|\cdots|\bm c_{r, \theta}(x_{r-1})|\bm w^\top\big)}{\theta V(x_1, x_2, \cdots, x_{r-1})}$ ,  becomes a quadratic polynomial in $x_{r-1}$ \textbf{},  which does not necessarily have a solution over 
$\mathbb{F}_q$ for arbitrary fixed $x_1, \cdots, x_{r-2}$. To overcome this challenge, we employ a subtle technique to eliminate the degree-one term of the quadratic polynomial, thereby ensuring that it always admits a solution over $\mathbb{F}_q$ when $q$ is even.
\end{remark}

\begin{lemma}[Schwartz-Zippel Lemma \cite{H06, S80,Z79}]\label{SZlem}
    Let $F(x_1, \cdots, x_{n})$ be a nonzero polynomial over $\mathbb{F}_q$. If $\deg_{x_i}(F)<q$, for any $1 \leq i \leq n$, then there exists $\alpha_1, \cdots, \alpha_n \in \mathbb{F}_q$  such that $F(\alpha_1, \cdots, \alpha_{n})\neq 0$.
\end{lemma}
 By Lemma \ref{lem5} and the Schwartz-Zippel Lemma, we obtain a characterization of the syndrome the deep hole of  ${\rm TRS}_k(\mathbb{F}_q, \theta)$ as follows.

\begin{proposition}\label{prop9}
  Let $q=2^m \geq 8$ and $\frac{3q-4}{4}\leq k \leq q-4$. Suppose $\bm u \in \mathbb{F}^q_q$ is a deep hole of ${\rm TRS}_k(\mathbb{F}_q, \theta)$ and its syndrome $\bm w^{\top}=H \cdot \bm u^{\top}$, then $\bm w=(0,0,\cdots, 0, w_{r-1})$ with $w_{r-1} \in \mathbb{F}^*_q$ or $\bm w=(0, \cdots, 0, w_{r-3}, w_{r-2}, w_{r-1})$ with $w_{r-3} \neq 0$ and $w_{r-3}\theta+w_{r-2}=0$
\end{proposition}

\begin{proof}
 By Lemma \ref{lem5}, the polynomial $P(x_1, \cdots, x_{r-3}) \in \mathbb{F}_q[x_1, x_2, \cdots, x_{r-3}]$ vanishes on $\underbrace{\mathbb{F}_q\times \mathbb{F}_q \times \cdots \times \mathbb{F}_q}_{r-3}.$ Note that $\deg_{x_i}(P)=(r-4)+2+(r-4)+(4+2(r-4))+4=4(q-k)-6<q$. By the Schwartz-Zippel Lemma, we have $P(x_1, \cdots, x_{r-3}) \equiv0$. We then divide our discussion into 3 cases.
		
\emph{Case 1}: $\tilde{f} \equiv 0$, i.e., $\tilde{f}=\sum\limits_{i=0}^{r-2}w_iS_i(x_1, \cdots,x_{r-3}, \sum\limits_{j=1}^{r-3}x_j+\theta^{-1})\equiv 0$. For each $0 \leq i \leq r-4$, the coefficient of the term $x_1\prod\limits_{j=1}^{r-3-i}x_{j}$ in $\sum\limits_{i=0}^{r-2}w_iS_i(x_1, \cdots,x_{r-3}, \sum\limits_{j=1}^{r-3}x_j+\theta^{-1})$ is equal to $w_i$,   which implies that $w_0=w_1=\cdots=w_{r-4}=0$ and $w_{r-3}\theta^{-1}+w_{r-2}=0$. If $w_{r-3} \neq 0,$ we are done.
If $w_{r-3}=0$, then $w_{r-2}=0$. We claim that $w_{r-1} \neq 0,$ otherwise, we have $\bm w= \bm 0$, which implies that $\bm u \in {\rm TRS}_k(\mathbb{F}_q, \theta)$, not a deep hole. Thus $\bm w=(0, 0, \cdots, 0, w_{r-1})$ with $w_{r-1} \in \mathbb{F}^*_q$.

\emph{Case 2}: $\tilde{f}x_i^2+\tilde{h}\equiv0$ for some $1 \leq i \leq r-3$. Note that $\tilde{f}$ and $\tilde{h}$ are symmetric polynomials. Thus $\tilde{f} =\tilde{h}\equiv 0$ or $r-3=1$. If $\tilde{f} =\tilde{h}\equiv 0$, then from Case 1, we have $w_0=w_1=\cdots=w_{r-4}=0$. Similarly, $\tilde{h}\equiv 0$ implies that $w_2=w_3=\cdots=w_{r-1}=0$, thus $\bm w=\bm 0$. If $r-3=1$, then
$$\begin{cases}
    \tilde{f}=w_0S_0(x_1,  x_{1}+\theta^{-1})+w_1 S_1(x_1,  x_{1}+\theta^{-1})+w_2S_2(x_1, x_{1}+\theta^{-1})\\
\hspace{0.35cm}=w_0x^2_1+\theta^{-1}w_0x_1+\theta^{-1}w_1+w_2,\\
\tilde{h}=w_2x^2_1+\theta^{-1}w_2x_1+w_3 \theta^{-1},\\ \tilde{f}x_1^2+\tilde{h} \equiv0.\end{cases}$$ It deduces that $w_0=w_1=w_2=w_3=0$. Therefore, in this case, we have $\bm w=\bm 0$. i.e., $\bm u \in {\rm TRS}_k(\mathbb{F}_q, \theta)$, which is not a deep hole.
		
\emph{Case 3}: {\small$\tilde{f}(\sum\limits_{j=1}^{r-3}x_j+\theta^{-1})^2+\tilde{h} \equiv 0.$ } Note that for $0 \leq j \leq r-4$, the coefficient of the term $x_1^2\cdot (\prod\limits_{\ell=1}^{r-3-j}x_{\ell})\cdot x_1$ in $\tilde{f} (\sum\limits_{j=1}^{r-3}x_j+\theta^{-1})^2+\tilde{h}$ is $w_{j}$,  thus $w_j=0$ for $j=0, 1, \cdots, r-4$. Combining with the definitions of $\tilde{f}$ and $\tilde{h}$, it holds that
\begin{eqnarray*}
   && (\sum\limits_{j=1}^{r-3}x^2_j+\theta^{-2})(w_{r-3}\theta^{-1}+w_{r-2})+w_{r-3}S_{r-5}(x_1, \cdots,x_{r-3},\sum\limits_{j=1}^{r-3}x_j+\theta^{-1})+\\
   &&w_{r-2}S_{r-4}(x_1, \cdots,x_{r-3},\sum\limits_{j=1}^{r-3}x_j+\theta^{-1})+\theta^{-1}w_{r-1}\equiv 0.
\end{eqnarray*} As a result, the coefficient of $x_1^2$ is $w_{r-3}\theta^{-1}$, which implies that $w_{r-3}=0$.  Then it can easily deduce that $w_{r-2}=w_{r-1}=0$. Hence we obtain that $\bm w=\bm 0$, i.e., $\bm u$ is a codeword, not a deep hole.
		
In summary, the conclusion follows. 
\end{proof} 
Furthermore, in the following proposition, we will prove that a vector with a syndrome of the form in the second case of Proposition \ref{prop9} actually cannot be a deep hole.  
\begin{proposition}\label{prop10}
    Suppose $q=2^m \geq 8$ and $\frac{3q-4}{4}\leq k \leq q-4$. Let $\bm u \in \mathbb{F}^q_q$ and 
    \[\bm w^{\top}=H \cdot \bm u^{\top}=(0, \cdots, 0, w_{r-3}, w_{r-2}, w_{r-1})\] with $w_{r-3} \neq 0$ and $w_{r-3}\theta+w_{r-2}=0$, then $\bm u$ is not a deep hole of ${\rm TRS}_k(\mathbb{F}_q, \theta)$.
\end{proposition}

\begin{proof}
   The proof is rather complex and involves results related to certain cubic Gauss sums, please refer to \ref{A}.
\end{proof}
Now, combining Proposition \ref{prop9}, Proposition \ref{prop10} with Corollary \ref{cor1}, we obtain our main result about the deep holes of ${\rm TRS}_k(\mathbb{F}_q, \theta)$ when $q$ is even and $\frac{3q-4}{4}\leq k \leq q-4$.
\begin{theorem}\label{even}
Suppose $q=2^m \geq 8$. If $\frac{3q-4}{4}\leq k \leq q-4$, then Theorem \ref{thm5} provides all the deep holes of ${\rm TRS}_k(\mathbb{F}_q, \theta)$.
\end{theorem}

	Finally, we determine all deep holes of ${\rm TRS}_k(\mathbb{F}_q, \theta)$ for even $q$ with $q \geq 8$ and $q-3 \leq k \leq  q-1$.

	\begin{theorem}\label{thm7}
		Suppose $q=2^m \geq 8$. Let $\bm u \in \mathbb{F}^q_q$ and $\bm w^\top =H \cdot \bm u^\top=(w_0,  \cdots, w_{r-1})^\top$, where $H$ is the parity-check matrix of ${\rm TRS}_k(\mathbb{F}_q, \theta)$ given in Eq. \eqref{pc}.
		
		\textnormal{\textbf{(i)}} For $k=q-1$, $\bm u$ is a deep hole of ${\rm TRS}_k(\mathbb{F}_q, \theta)$ if and only if $\bm u$ is generated by Theorem \ref{thm5}.
		
		\textnormal{\textbf{(ii)}} For $k=q-2$, $\bm u$ is a deep hole of ${\rm TRS}_k(\mathbb{F}_q, \theta)$ if and only if $\bm u$ is generated by  $w_0x^{q-1}+w_1x^{q-2}+f_{q-2,\theta}(x)$ with $w_0=0, w_1 \neq 0$ or $w_0 \neq 0$ and ${\rm Tr}(\frac{w_1\theta}{w_0})=1$, where $f_{q-2,\theta}(x) \in \mathcal{S}_{q-2, \theta}$ and ${\rm Tr}(x)$ is the trace function from $\mathbb{F}_q$ to $\mathbb{F}_2$.
		
		\textnormal{\textbf{(iii)}} For $k=q-3$, $\bm u$ is a deep hole of ${\rm TRS}_k(\mathbb{F}_q, \theta)$ if and only if $\bm u$ is given by  Theorem \ref{thm5} or generated by $a(x^{q-2}+\theta^{-1}x^{q-3})+f_{q-3,\theta}(x)$ with $a \neq 0$, $f_{q-3,\theta}(x) \in \mathcal{S}_{q-3, \theta}$ and $2 \nmid m$.
	\end{theorem}
	\begin{proof}
		\textbf{(i):} For $k=q-1$,  then $\rho({\rm TRS}_k(\mathbb{F}_q, \theta))=q-k=1$. Thus every non-codeword is a deep hole. The conclusion can be easily verified.
		
		\textbf{(ii): }For $k=q-2$, then by Proposition \ref{prop8}, $\bm u$ is a deep hole of ${\rm TRS}_k(\mathbb{F}_q, \theta)$ if and only if 	for any $\alpha \in \mathbb{F}_q$. 
		\begin{equation}\label{3}
			\det\left(\begin{array}{cc}
				1  & w_0 \\
				\alpha+\theta \alpha^{2} &w_1\\
			\end{array}\right)=\theta w_0\alpha^2+w_0 \alpha +w_1 \neq0,
		\end{equation}
		where $(w_0,w_1)^\top=H{\bm u}^\top$.
		Hence, Eq. \eqref{3} holds $\Longleftrightarrow 
		w_0 =0,w_1 \neq 0$, or,  $w_0 \neq 0$ and 
		$\theta w_0x^2+w_0x +w_1=0$ has no roots in $\mathbb{F}_q$. By \cite[Corollary 3.79]{LN97}, the latter one is equivalent to ${\rm Tr}(\frac{w_1\theta}{w_0})=1$.  
		
		Furthermore, given $w_0, w_1$ with $w_0=0, w_1 \neq 0$ or $w_0 \neq 0$ and ${\rm Tr}(\frac{w_1\theta}{w_0})=1$,  it can be verified that the vector $\bm u_f$ generated by $f(x)=w_0x^{q-1}+w_1x^{q-2}+f_{q-2,\theta}(x)$ satisfies $H\cdot \bm u^\top_f=(w_0, w_1)^\top$. Thus such vectors are all deep holes of ${\rm TRS}_{q-2}(\mathbb{F}_q, \theta)$.
		
		\textbf{(iii): } For $k=q-3$, then by Proposition \ref{prop8}, $\bm u$ is a deep hole of ${\rm TRS}_k(\mathbb{F}_q, \theta)$ if and only if 
		\[D(\alpha,\beta)=\frac{\det\left(\begin{array}{ccc}
				1  & 1 & w_0 \\
				\alpha & \beta &w_1\\
				\alpha^2+\theta \alpha^{3} &\beta^2+\theta \beta^{3} &w_2\\
			\end{array}\right)}{\alpha+\beta}\neq0,\]
		for any $\alpha \neq \beta \in \mathbb{F}_q$,  where $\bm w^\top=(w_0,w_1,w_2)^\top=H{\bm u}^\top$. W.l.o.g., we suppose $\alpha \neq 0$ and $\beta=\lambda \alpha$, then 
		\begin{eqnarray}\label{4}
		    \tilde{D}(\alpha,\lambda)&\triangleq &D(\alpha,\lambda\alpha)\nonumber\\&=&(w_0\alpha+w_1)\theta\alpha^2\lambda^2 +(w_0\alpha+w_1)\cdot \alpha(1+\theta\alpha)\lambda+w_2\nonumber\\
       &&+w_1\alpha(1+\theta\alpha) \neq 0, \text{ for any }\alpha \neq 0 \in \mathbb{F}_q \text{ and } \lambda \neq 1 \in \mathbb{F}_q.
		\end{eqnarray}
Next, we determine all vectors $\bm w^\top\in\mathbb{F}_q^3\setminus\{\bm 0\}$  that satisfy (\ref{4}).
		
		\emph{Case 1}: $w_0=0$ and $w_1=0$. Then $\tilde{D}(\alpha,\lambda)=w_2 \neq 0$, i.e.,  (\ref{4}) holds.
		
		\emph{Case 2}:  $w_0=0$ and $w_1 \neq 0$. Then (\ref{4})  holds if and only if $w_2=\theta^{-1}w_1$ with $2\nmid m$. 
  
		When $w_2=0$, it has $w_2=\tilde{D}(\theta^{-1},0)=0$. When  $w_2\neq 0$ and $2\mid m$, then $\tilde{D}(\frac{w_2\omega}{w_1},\omega)=0$,
		where $\omega$ is the primitive cubic root of unity. When $w_2\neq 0$ and $2\nmid m$, (\ref{4}) holds if and only if  $\frac{w_2}{w_1} \neq \theta\alpha^2\lambda^2+\alpha(1+\theta \alpha)(\lambda+1)$ for any $\alpha \neq 0$ and $\lambda \neq 1$. Let $\alpha=\theta^{-1}$, then $\frac{w_2}{w_1} \neq \frac{\lambda^2}{\theta}$ for any $\lambda\neq 1$. Thus (\ref{4}) holds implies that $\frac{w_2}{w_1}=\theta^{-1}$. Conversely,  let $\frac{w_2}{w_1}=\theta^{-1}$, we show that for any $\alpha \neq 0$ and  $\lambda \neq 1$, $\theta\alpha^2\lambda^2+\alpha(1+\theta \alpha)(\lambda+1) \neq \theta^{-1}$. By contradiction, if $\theta\alpha^2\lambda^2+\alpha(1+\theta \alpha)(\lambda+1) = \theta^{-1}$ for some $\alpha \neq 0$ and  $\lambda \neq 1 $. Then $1+\theta \alpha \neq 0$, otherwise, we can deduce from the equation that $\lambda=1$. Then we obtain that $$ \theta^2\alpha^2\lambda^2+\theta\alpha(1+\theta \alpha)\lambda+\theta\alpha(1+\theta \alpha) +1=0,$$ i.e., $$(\frac{\theta \alpha \lambda}{1+\theta\alpha})^2+\frac{\theta \alpha \lambda}{1+\theta\alpha}+\frac{1+\theta \alpha +\theta^2\alpha^2}{1+\theta^2\alpha^2}=0.$$  Hence, we have ${\rm Tr}(\frac{1+\theta \alpha +\theta^2\alpha^2}{1+\theta^2\alpha^2})={\rm Tr}(1+\frac{\theta \alpha }{1+\theta^2\alpha^2})=0.$  Note that  ${\rm Tr}(\frac{\theta \alpha }{1+\theta^2\alpha^2})=$ ${\rm Tr}(\frac{1 }{1+\theta\alpha}+\frac{1 }{(1+\theta\alpha)^2})=0$, thus, $0={\rm Tr}(1+\frac{\theta \alpha }{1+\theta^2\alpha^2})={\rm Tr}(1)$. However, ${\rm Tr}(1)=1$ since $2 \nmid m$, which leads to a contradiction.
		
		\emph{Case 3}: $w_0 \neq 0$ and $w_1=0$. Then (\ref{4}) does not hold.
  
  Otherwise,  $\frac{w_2}{w_0} \neq \lambda \alpha^2(1+\theta\alpha(1+\lambda))$ for any $\alpha \neq 0$ and $\lambda \neq 1$. Choose $\alpha=\theta^{-1}$, then $\frac{w_2}{w_0} \neq \frac{\lambda^2}{\theta^2}$ for any $\lambda \neq 1$. Thus $\frac{w_2}{w_0} =\frac{1}{\theta^2}$. Next we will show that there exist $\alpha \neq 0$ and $\lambda \neq 1$ such that $
		\lambda \alpha^2(1+\theta\alpha(1+\lambda))=\frac{1}{\theta^2}$
		which leads to a contradiction. The above equation is equivalent to 
		$(\frac{\theta \alpha \lambda}{1+\theta\alpha})^2+\frac{\theta \alpha \lambda}{1+\theta\alpha}+\frac{1}{(1+\theta \alpha)^2\theta\alpha}=0.$
		By \cite[Corollary 3.79]{LN97}, it is equivalent to ${\rm Tr}(\frac{1}{(1+\theta \alpha)^2\theta\alpha})=0$. Note that $${\rm Tr}(\frac{1}{(1+\theta \alpha)^2\theta\alpha})={\rm Tr}(\frac{1}{\theta\alpha}+\frac{1}{1+\theta\alpha}+(\frac{1}{1+\theta\alpha})^2)={\rm Tr}(\frac{1}{\theta\alpha}).$$ Thus when $q \geq 8$, we can always choose $\alpha \in \mathbb{F}_q\backslash\{0, \theta^{-1}\}$ such that ${\rm Tr}(\frac{1}{\theta\alpha})=0$.
		
		\emph{Case 4}:  $w_0 \neq 0$ and $w_2= 0$. Then we can find that $\tilde{D}(\theta^{-1},0)=0$, which implies that (\ref{4}) does not hold.

		\emph{Case 5}: $w_0, w_1, w_2 \neq 0$. Then (\ref{4}) does not hold. 
  
  W.l.o.g., we suppose $w_0=1$. We will show that there always exists $\alpha \neq 0$ and $\lambda \neq 1$ such that $\tilde{D}(\alpha, \lambda)=0$. 
		
		If $w_1 \neq \theta^{-1}$ and $w_2 \neq \theta^{-1}(w_1+\theta^{-1})$, we choose $\alpha=\theta^{-1}$, then $\tilde{D}(\alpha, \lambda)=\theta^{-1}(w_1+\theta^{-1})\lambda^2+w_2$. Obviously, there exists $\lambda \neq 1 \in \mathbb{F}_q$ such that $\tilde{D}(\alpha, \lambda)=0$.
		
		If $w_1 \neq \theta^{-1}$ and $w_2 = \theta^{-1}(w_1+\theta^{-1})$, then $\tilde{D}(\alpha, \lambda)=(\alpha+w_1)\theta\alpha^2\lambda^2+(\alpha+w_1)\alpha(1+\theta\alpha)\lambda +w_1\alpha(1+\theta\alpha)+\theta^{-1}(w_1+\theta^{-1}).$ We suppose $\alpha \neq 0, w_1, \theta^{-1}$, then by \cite[Corollary 3.79]{LN97}, $\tilde{D}(\alpha, \lambda)=0$ has a solution $\lambda \neq 1 \in \mathbb{F}_q$ if and only if 
		\begin{eqnarray*}
			&&{\rm Tr}(\frac{w_1\theta\alpha(1+\theta\alpha)+(w_1+\theta^{-1})}{(\alpha+w_1)(1+\theta\alpha)^2})\\
			&=&{\rm Tr}(\frac{\theta^{-1}(1+w_1\theta)(1+\theta^2\alpha^2)+\theta\alpha(w_1+\alpha)}{(\alpha+w_1)(1+\theta\alpha)^2})\\
			&=& {\rm Tr}(\frac{1+w_1\theta}{\theta(\alpha+w_1)})+{\rm Tr}(\frac{1}{1+\theta \alpha})+{\rm Tr}(\frac{1}{(1+\theta\alpha)^2}) \\
			&=&{\rm Tr}(\frac{1+w_1\theta}{\theta(\alpha+w_1)})=0.
		\end{eqnarray*}
		When $q \geq 8$, we can choose $\alpha \in \mathbb{F}_q\backslash\{0, w_1, \theta^{-1}\}$ such that ${\rm Tr}(\frac{1+w_1\theta}{\theta(\alpha+w_1)})=0$. 
		
		If $w_1 = \theta^{-1}$, then $\tilde{D}(\alpha, \lambda)=\alpha^2(1+\theta \alpha)\lambda^2+\frac{\alpha(1+\theta\alpha)^2}{\theta}\lambda+\frac{\alpha(1+\theta\alpha)}{\theta}+w_2$. When ${\rm Tr}(w_2\theta^2)=0$, then we can choose $\alpha \in \mathbb{F}_q$ such that $\frac{\alpha(1+\theta\alpha)}{\theta}+w_2=0$ and $\tilde{D}(\alpha, 0)=0$. When ${\rm Tr}(w_2\theta^2)=1$, we suppose $\alpha \neq 0, \theta^{-1}$, then $\tilde{D}(\alpha, \lambda)=0$ has a solution $\lambda \neq 1 \in \mathbb{F}_q$ if and only if 
		\begin{align*}    
		&{\rm Tr}(\frac{\alpha\theta(1+\theta\alpha)+w_2\theta^{2}}{(1+\theta\alpha)^3})\\
  =&{\rm Tr}(\frac{\theta \alpha}{(1+\theta\alpha)^2})+{\rm Tr}(\frac{w_2\theta^{2}}{(1+\theta\alpha)^3})\\ =&{\rm Tr}(\frac{w_2\theta^{2}}{(1+\theta\alpha)^3}) =0.\end{align*}
		Let $\gamma=\frac{1}{1+\theta \alpha}$, then $\gamma \in \mathbb{F}_q\backslash\{0,1\}$. If $3 \nmid q-1$, then  $x^3$ is a permutation polynomial of $\mathbb{F}_q$. Thus we can always choose $\gamma \in \mathbb{F}_q\backslash\{0,1\}$ such that  ${\rm Tr}(w_2\theta^2 \gamma^3)=0$ for $q\geq 8$. If $3 \mid q-1$, to obtain a  contradiction, we suppose ${\rm Tr}(w_2\theta^2 \gamma^3)=1$ for any $\gamma \in \mathbb{F}_q\backslash\{0,1\}$.  Then 
		\begin{eqnarray*}
			-(q-2)&=&\sum_{\gamma \in \mathbb{F}_q\backslash\{0,1\}}(-1)^{{\rm Tr}(w_2\theta^2 \gamma^3)}\\
			& =&\sum_{\gamma \in \mathbb{F}_q}(-1)^{{\rm Tr}(w_2\theta^2 \gamma^3)}-1-(-1)^{{\rm Tr}(w_2\theta^2)}\\
			&=&\sum_{\gamma \in \mathbb{F}_q}(-1)^{{\rm Tr}(w_2\theta^2 \gamma^3)}.
		\end{eqnarray*}
		By Proposition \ref{prop2},  we have $|\sum\limits_{\gamma \in \mathbb{F}_q}(-1)^{{\rm Tr}(w_2\theta^2\gamma^3)}| \leq 2\sqrt{q}$, which leads to $q-2 \leq 2\sqrt{q} \Longrightarrow q <8$, contradiction.

		In short, $\bm u$ is a deep hole of ${\rm TRS}_k(\mathbb{F}_q, \theta)$ if and only if $\bm w=a(0, 0, 1)$ or $\bm w=a(0, 1, \theta^{-1})$ when $2 \nmid m$, where $a \in \mathbb{F}^*_q$. The conclusion then follows from Proposition \ref{prop7}.  
	\end{proof}
	
\subsection{Deep Holes of ${\rm TRS}_k(\mathbb{F}_q, \theta)$ for Odd $q$}
In Subsection \ref{sec4.2}, a key point of our proof was that the determinant 
\[\frac{\det\big(\bm c_{r, \theta}(x_1)|\bm c_{r, \theta}(x_2)|\cdots|\bm c_{r, \theta}(x_{r-1})|\bm w^\top\big)}{\theta V(x_1, x_2, \cdots, x_{r-1})}\] can be viewed as a quadratic polynomial in $x_{r-1}$, which has a solution under certain conditions when $q$ is even. However, this approach does not apply to the case where $q$ is odd. Therefore, in this subsection, we will adopt a different method to handle it. For odd $q$, we will determine all deep holes of ${\rm TRS}_k(\mathbb{F}_q, \theta)$ when  $\frac{3q+3\sqrt{q}-7}{4}\leq k\leq q-1$.

Firstly, we treat the case of $\frac{3q+3\sqrt{q}-7}{4}\leq k\leq q-4$ and prove that the standard deep holes are all the deep holes of ${\rm TRS}_k(\mathbb{F}_q, \theta)$ by using Corollary \ref{cor1} and Proposition \ref{prop8}. We start with two lemmas which will be used in the proof of our Theorem \ref{thm9}.

The following lemma shows that any vector whose syndrome is of the form $(0, 0, \cdots, 0, 1, \lambda) \in \mathbb{F}^r_q$ is not a deep hole of ${\rm TRS}_k(\mathbb{F}_q, \theta)$. 
\begin{lemma}\label{lem8}
Suppose $q$ is odd and $4\leq r \leq \frac{q+8}{4}$. For any $\lambda \in \mathbb{F}_q$, denote $\bm v=(0, 0, \cdots, 0, 1, \lambda) \in \mathbb{F}^r_q$. Then there exist pairwise distinct $r-1$ elements $\alpha_1, \alpha_2, \cdots, \alpha_{r-1} \in \mathbb{F}_q$ such that 
\[\det\big(\bm c_{r, \theta}(\alpha_1)|\bm c_{r, \theta}(\alpha_2)|\cdots|\bm c_{r, \theta}(\alpha_{r-1})|\bm v^\top\big)=0.\]
\end{lemma}
\begin{proof}
 By Lemmas \ref{det2} and \ref{det3}, we have \begin{eqnarray*}
     &&\det\big(\bm c_{r, \theta}(\alpha_1)|\bm c_{r, \theta}(\alpha_2)|\cdots|\bm c_{r, \theta}(\alpha_{r-1})|\bm v^\top\big)\\
  &=& \lambda  \det \left(\begin{array}{ccc}
		1 & \cdots & 1 \\
		\alpha_1 & \cdots & \alpha_{r-1} \\
		\vdots & \ddots & \vdots \\
		\alpha_1^{r-2} & \cdots & \alpha_{r-1}^{r-2} \\
		\end{array}\right)- \det \left(\begin{array}{ccc}
		1 & \cdots & 1 \\
		\alpha_1 & \cdots & \alpha_{r-1} \\
		\vdots & \ddots & \vdots \\
		\alpha_1^{r-3} & \cdots & \alpha_{r-1}^{r-3} \\
		\alpha_1^{r-1} & \cdots &\alpha_{r-1}^{r-1} \\
		\end{array}\right) \\
  &&+\theta \cdot \det \left(\begin{array}{ccc}
		1 & \cdots & 1 \\
		\alpha_1 & \cdots & \alpha_{r-1} \\
		\vdots & \ddots & \vdots \\
		\alpha_1^{r-3} & \cdots & \alpha_{r-1}^{r-3} \\
		\alpha_1^{r} & \cdots &\alpha_{r-1}^{r} \\
		\end{array}\right)\\
  &=&\theta V(\alpha_1, \cdots, \alpha_{r-1})\big(\theta^{-1}\lambda+ \sum\limits_{1 \leq i \leq j \leq r-1}\alpha_i \alpha_j-\theta^{-1}\sum\limits_{i=1}^{r-1}\alpha_i \big).\\
 \end{eqnarray*}
Denote $\beta_0=\sum\limits_{1 \leq i \leq j \leq r-4}\alpha_i \alpha_j,\beta_1=\sum\limits_{i=1}^{r-4}\alpha_i, b=-\alpha^2_{r-3}+(\theta^{-1}-\beta_1)\alpha_{r-3}+(\theta^{-1}\beta_1-\theta^{-1}\lambda -\beta_0),X=\frac{\alpha_{r-2}+\alpha_{r-1}}{2},$ and $Y=\frac{\alpha_{r-2}-\alpha_{r-1}}{2}$. Then 
\[\theta^{-1}\lambda+ \sum\limits_{1 \leq i \leq j \leq r-1}\alpha_i \alpha_j-\theta^{-1}\sum\limits_{i=1}^{r-1}\alpha_i =3X^2+Y^2-2(\theta^{-1}-\beta_1-\alpha_{r-3})X-b.\]
Thus it is sufficient to show that there exists $\alpha_{r-3}\in \mathbb{F}_q\setminus\{\alpha_1,\cdots,\alpha_{r-4}\}$ such that  the equation  \begin{equation}\label{r1}
 F(X,Y)=3X^2+Y^2-2(\theta^{-1}-\beta_1-\alpha_{r-3})X-b=0   
 \end{equation} has  a solution $(X,Y)\in\mathbb{F}_q^2$ and  $\alpha_1,\cdots,\alpha_{r-4},\alpha_{r-3},X+Y,X-Y$ are all distinct.

 Note that ${\rm N}(F(X,X-\alpha_i))\leq 2, {\rm N}(F(X,\alpha_i-X))\leq 2$ for all $1\leq i\leq r-3 $ and ${\rm N}(F(X,0))\leq 2$, so we only need to show that there exists  $\alpha_{r-3}\in \mathbb{F}_q\setminus\{\alpha_1,\cdots,\alpha_{r-4}\}$ such that 
${\rm N}(F(X,Y))>4(r-3)+2=4r-10$. 

When $\rm Char({\mathbb{F}_q})=3$, we choose $\alpha_{r-3}\in \mathbb{F}_q\setminus\{\alpha_1,\cdots,\alpha_{r-4},\theta^{-1}-\beta_1\}$. Combining with (\ref{r1}), it has ${\rm N}(F(X,Y))=q\geq 4r-8>4r-10$. 
When $\rm Char({\mathbb{F}_q})>3$, $|\mathbb{F}_q\setminus\{\alpha_1,\cdots,\alpha_{r-4}\}|=q-r+4>2$, so we may choose $\alpha_{r-3}\in \mathbb{F}_q\setminus\{\alpha_1,\cdots,\alpha_{r-4}\}$   such that $\tilde{b}\triangleq b+\frac{(\theta^{-1}-\beta_1-\alpha_{r-3})^2}{3}\neq 0$. Moreover, $F(X,Y)$ in equation (\ref{r1}) can be rewritten as $$F(X,Y)=3(X-\frac{\theta^{-1}-\beta_1-\alpha_{r-3}}{3})^2+Y^2-\tilde{b}.$$ By 
Proposition \ref{prop4}, it has that ${\rm N}(F(X,Y))\geq q-1\geq 4r-9>4r-10$.\end{proof}

The following lemma will be used in the second case of the proof of Theorem \ref{thm9}, which considers the vector whose syndrome is of the form $(\bm w', \lambda) \in \mathbb{F}^r_q$ with $\bm w'^{\top} \neq a\bm c_{r-1}(\infty)$ for any $a \in \mathbb{F}_q$.
\begin{lemma}\label{lem9}
 Suppose $3 \leq r \leq \frac{q-3\sqrt{q}+7}{4}$. Given any  subset $\{\alpha_1, \alpha_2, \cdots, \alpha_{r-2}\}  \subseteq \mathbb{F}_q$ and $b \in \mathbb{F}_q$, there exist $\alpha_{r-1} \neq \alpha_{r} \in \mathbb{F}_q\backslash \{\alpha_1, \alpha_2, \cdots, \alpha_{r-2}\}$ such that 
	\[\det\big(\bm c_{r, \theta}(\alpha_1)+b \bm c_{r}(\infty)|\bm c_{r, \theta}(\alpha_2)| \cdots|\bm c_{r, \theta}(\alpha_r)\big)=0.\]
\end{lemma}
\begin{proof}
By Lemma \ref{det3}, we have
	\begin{eqnarray*}
	&&\det\big(\bm c_{r, \theta}(\alpha_1)+b \bm c_{r}(\infty)|\bm c_{r, \theta}(\alpha_2)| \cdots|\bm c_{r, \theta}(\alpha_r)\big)\\
	&=&\det\big(\bm c_{r, \theta}(\alpha_1)|\bm c_{r, \theta}(\alpha_2)| \cdots| \bm c_{r, \theta}(\alpha_r)\big)+b\det\big(\bm c_{r}(\infty)|\bm c_{r, \theta}(\alpha_2)| \cdots|\bm c_{r, \theta}(\alpha_r)\big)\\
	&=&V(\alpha_1, \alpha_2, \cdots, \alpha_r)(1-\theta\sum_{i=1}^{r}\alpha_i)+bV(\alpha_2, \alpha_3, \cdots, \alpha_r)\\
	&=&V(\alpha_2, \alpha_3, \cdots, \alpha_{r})((1-\theta\sum_{i=1}^{r}\alpha_i)\prod_{j=2}^{r}(\alpha_j-\alpha_1)+b).\\
	\end{eqnarray*}
	Denote $\beta_0=\prod\limits_{i=2}^{r-2}(\alpha_i-\alpha_1)$,  $\xi=1-\theta\sum\limits_{i=2}^{r-2}\alpha_i-2\theta\alpha_1$, $X=\alpha_{r-1}-\alpha_1$, and   $Y=\alpha_r-\alpha_1$. Then we only need to show that the equation
	\begin{equation}\label{5}
	F(X, Y)\triangleq\beta_0XY(\xi-\theta(X+Y))+b=0,
	\end{equation}
	which has a solution $(X, Y) \in \mathbb{F}^2_q$ with $X \neq Y\in \mathbb{F}_q \backslash \mathcal{S}$, where $\mathcal{S}=\{0, \alpha_2-\alpha_1, \cdots, \alpha_{r-2}-\alpha_1\}$. 
	
	If $b=0$, then it is equivalent to show that $F_0(X,Y)\triangleq\theta(X+Y)-\xi=0$ 	has a solution $(X, Y) \in \mathbb{F}^2_q$ with $X \neq Y\in \mathbb{F}_q \backslash \mathcal{S}$. Note that ${\rm N}(F_0(X,Y))=|\{(X, Y) \in \mathbb{F}^2_q: F_0(X,Y)=0\}|=q$. Given $X \in \mathcal{S}$, then $|\{Y \in \mathbb{F}_q: F_0(X,Y)=0\}| =1$ and $|\{(X, Y) \in \mathbb{F}^2_q: F_0(X,Y)=0\textnormal{ and } X=Y\}| =1$. The conclusion then follows from $N_0=q >2\times |\mathcal{S}|+1=2r-1$.
	
For $b \neq 0$, then the solution of Eq. \eqref{5} satisfies that $(X, Y) \in (\mathbb{F}^*_q)^2$.  For each $\beta \in \mathcal{S} \backslash \{0\}$,
 ${\rm N}(F(X,\beta))={\rm N}(F(\beta,X))\leq 2$. 
 In addition, ${\rm N}(F(X,X))\leq 3$. Thus, we only need to show that ${\rm N}(F(X,Y))>2\times 2 \times (|\mathcal{S}|-1)+3=4r-9$, i.e., ${\rm N}(F(X,Y)) \geq 4r-8$. Next, we use character sums to estimate the value of ${\rm N}(F(X,Y))$.

Let $\chi(x)$ be the canonical additive character of $\mathbb{F}_q$ and  $\eta$ be the quadratic character of $\mathbb{F}_q$. Then \begin{eqnarray*}
	{\rm N}(F(X,Y))&=& \frac{1}{q}\sum_{X, Y \in \mathbb{F}_q}\sum_{z \in \mathbb{F}_q}\chi(zF(X,Y)) \\
	&=&q-1+\frac{1}{q}\sum_{Y,z \in \mathbb{F}^*_q}\sum_{X \in \mathbb{F}_q}\chi(z\beta_0\theta Y X^2+z\beta_0 Y(\theta Y-\xi)X-zb) \\
	&\stackrel{\rm(a)}{=}&q-1+\frac{1}{q}\sum_{Y, z \in \mathbb{F}^*_q}\chi(zb+\frac{z\beta_0Y(\theta Y-\xi)^2}{4\theta})\eta(-z\beta_0\theta Y)G(\eta, \chi) \\
&\stackrel{\rm(b)}{=}& q-1+\frac{G(\eta, \chi)}{q}\sum_{W \in \mathbb{F}^*_q}\eta(W)\sum_{z \in \mathbb{F}^*_q}\chi(zf(W))\eta(z),\\
\end{eqnarray*}
where (a) follows from Proposition \ref{prop2} (ii), (b) is because that $\beta_0\theta\neq 0$ and  $f(W)=b-\frac{W(\beta_0\xi+W)^2}{4\theta^2\beta^2_0}$ 
with  $W=-\beta_0 \theta Y$.  

Denote $\mathcal{Z}_0=\{W\in \mathbb{F}_q^*:f(W)=0\}$, then by Proposition \ref{gauss} and $\sum\limits_{z\in\mathbb{F}^*_q}\eta(z)=0$, we have 
\begin{eqnarray*}
	&&{\rm N}(F(X,Y)) \\
 &=& q-1+\frac{G(\eta, \chi)}{q}\big(\sum_{W \in \mathcal{Z}_0}\eta(W)\sum_{z \in \mathbb{F}^*_q}\eta(z)+\sum_{W \in \mathbb{F}^*_q\backslash \mathcal{Z}_0}\eta(W)\sum_{z \in \mathbb{F}^*_q}\chi(zf(W))\eta(z)\big)\\
 &=&q-1+\frac{G^2(\eta, \chi)}{q}\sum_{W\in \mathbb{F}^*_q\setminus\mathcal{Z}_0}\eta(Wf(W))\\
	&=&q-1+\frac{G^2(\eta, \chi)}{q}\sum_{W \in \mathbb{F}^*_q}\eta(Wf(W)).
\end{eqnarray*}
If $Wf(W)=g^2(W)$ for some non-zero polynomial $g(W) \in \mathbb{F}_q[W]$, then $g(0)=0$, write $g(W)=W\cdot h(W)$ for some nonzero polynomial $h(W) \in \mathbb{F}_q[W]$. Then $f(W)=Wh^2(W)$, which implies that  $b=f(0)=0$, a contradiction.  Thus, we have that $Wf(W)$ can not be a square of some polynomial in $\mathbb{F}_q[W]$. By Proposition \ref{prop3},	 it has 
$$\left|\sum_{W \in \mathbb{F}^*_q}\eta(Wf(W))\right|\leq 3\sqrt{q}.$$
Therefore, ${\rm N}(F(X,Y)) \geq q-1-\frac{|G(\eta, \chi)|^2}{q}|\sum\limits_{W \in \mathbb{F}^*_q}\eta(Wf(W))| \geq q-1-3\sqrt{q}$. 
Since $r \leq \frac{q-3\sqrt{q}+7}{4}$, we can deduce that ${\rm N}(F(X,Y)) \geq q-1-3\sqrt{q} \geq 4r-8$. The proof is completed.\end{proof}

 Seroussi and Roth \cite{SR86} proved the following nice result for the MDS extensions of RS codes.

\begin{theorem}[{\cite[Theorem 1]{SR86}}] \label{roth}
Suppose $\mathcal{A}=\{\alpha_1, \alpha_2, \cdots, \alpha_n\} \subseteq \mathbb{F}_q$ and $2 \leq s \leq n-\frac{q-1}{2}$. Let $G_s=\big(\bm c_{s}(\alpha_1)|\bm c_{s}(\alpha_2)|\cdots |\bm c_s(\alpha_n)\big)$ be the generator matrix of the Reed-Solomon code ${\rm RS}_s(\mathcal{A})$. Suppose $\bm w  \in \mathbb{F}_q^{s}$ be a vector. Then the augmented matrix $\big(G_s | \bm w^{\top}\big)$ generates an $[n+1, s]_q$ MDS code if and only if:
	
1) (For $q$ odd ) $\bm w^{\top}=a \bm c_{s}(\delta)$ for some $\delta \in (\mathbb{F}_q \cup \infty) \backslash \mathcal{A}$ and $a \in \mathbb{F}^*_q$;
	
2) (For $q$ even) $\bm w$ is either as above, or additionally in case $s=3, \bm w=(0, 0, a)$ for some $a \in \mathbb{F}^*_q$.
\end{theorem}

Applying the above lemmas and Seroussi and Roth's result, we determine the deep holes of  ${\rm TRS}_k(\mathbb{F}_q, \theta)$ for odd $q$.
\begin{theorem}\label{thm9}
	Suppose $q$ is odd. If $\frac{3q+3\sqrt{q}-7}{4} \leq k \leq q-4$, then  Theorem \ref{thm5} provides all deep holes of ${\rm TRS}_k(\mathbb{F}_q, \theta)$.
\end{theorem}

\begin{proof}
Denote $r=q-k$, then $4 \leq r \leq \frac{q-3\sqrt{q}+7}{4}$.	Suppose $\bm u$ is a deep hole of ${\rm TRS}_k(\mathbb{F}_q, \theta)$ and $\bm w^{\top}=H \cdot \bm u^{\top}$. Write $\bm w=(\bm w', \lambda),$ where $\bm w' \in \mathbb{F}^{r-1}_q$. 
	
(i): If $\bm w'^{\top}=a\bm c_{r-1}(\infty)$ for some $a \in \mathbb{F}^*_q$, then by Lemma \ref{lem8}, there exist distinct elements $\alpha_1, \alpha_2, \cdots, \alpha_{r-1} \in \mathbb{F}_q$ such that $\bm c_{r, \theta}(\alpha_1), \bm c_{r, \theta}(\alpha_2), \cdots, \bm c_{r, \theta}(\alpha_{r-1})$ and $H\cdot \bm u^\top$ are $\mathbb{F}_q$-linearly dependent, which contradicts Proposition \ref{prop8}.
	
(ii): If $\bm w'^{\top}\neq a\bm c_{r-1}(\infty)$ for any $a \in \mathbb{F}_q$, by Theorem \ref{roth}, there exist distinct elements $\alpha_1, \alpha_2, \cdots, \alpha_{r-2} \in \mathbb{F}_q$ such that $\bm c_{r-1}(\alpha_1),\bm c_{r-1}(\alpha_2),\cdots, \bm c_{r-1}(\alpha_{r-2})$ and $\bm w'^{\top}$ are linearly dependent, i.e., there exist not all zero elements $a_1, \cdots, a_{r-2} \in \mathbb{F}_q$, such that
$\bm w'=\sum_{i=1}^{r-2}a_i\bm c_{r-1}(\alpha_i).$
 Thus
\begin{eqnarray*}
		\bm w ^{\top}&=& \left(\begin{array}{c}
	   \sum\limits_{i=1}^{r-2}a_i\bm c_{r-1}(\alpha_i) \\
		   0
		\end{array}\right)+\lambda \bm c_r(\infty)\\
  &=& \sum_{i=1}^{r-2}a_i\left(\begin{array}{c}
			\bm c_{r-1}(\alpha_i) \\
			0
		\end{array}\right)+\lambda \bm c_r(\infty)\\
		&=& \sum_{i=1}^{r-2}a_i\bm c_{r, \theta}(\alpha_i)+b \bm c_r(\infty),\\
\end{eqnarray*}
where $b=-\sum_{i=1}^{r-2}a_i(\alpha^{r-1}_i-\theta \alpha^r_i)+\lambda$.
Without loss of generality, we suppose $a_1=1$. Then by Lemma \ref{lem9}, there exist not all zero elements $b_1, \cdots, b_{r-2} \in \mathbb{F}_q$ and $\alpha_{r-1} \neq \alpha_{r} \in \mathbb{F}_q\backslash \{\alpha_1, \alpha_2, \cdots, \alpha_{r-2}\}$ , such that
\[\bm c_{r, \theta}(\alpha_1)+b\bm c_{r}(\infty)=\sum_{i=2}^{r}b_i\bm c_{r, \theta}(\alpha_i).\]
Hence $\bm w^{\top}=\sum\limits_{i=2}^{r}(a_i+b_i)\bm c_{r, \theta}(\alpha_i)$, which contradicts to Proposition \ref{prop8}.

Therefore, we deduce that $\bm w'=\bm 0$, i.e., $\bm w=(0, 0, \cdots, 0, \lambda)$. If $\lambda=0$, then $\bm w= \bm 0$, hence $\bm u \in {\rm TRS}_k(\mathbb{F}_q, \theta)$, which is not a deep hole. So we have $\lambda \neq 0.$ Then the conclusion follows from Corollary \ref{cor1}.  
\end{proof}

	Finally, we determine all deep holes of ${\rm TRS}_k(\mathbb{F}_q, \theta)$ for odd $q>16$ and  $k\in \{q-3,q-2,q-1\}$.

	\begin{theorem}\label{thm10}
		Suppose $q$ is odd with $q>16$. Let $\bm u \in \mathbb{F}^q_q$ and $\bm w^\top =H \cdot \bm u^\top=(w_0, \cdots, w_{r-1})^\top$, where $H$ is the parity-check matrix of ${\rm TRS}_k(\mathbb{F}_q, \theta)$ given in Eq. \eqref{pc}.
		\begin{itemize}
		    \item [\textbf{(i)}]  When $k=q-1$, $\bm u$ is a deep hole of ${\rm TRS}_k(\mathbb{F}_q, \theta)$ if and only if $\bm u$ is generated by Theorem \ref{thm5}.
		
		\item [\textbf{(ii)}] When $k=q-2$, $\bm u$ is a deep hole of ${\rm TRS}_k(\mathbb{F}_q, \theta)$ if and only if $\bm u$ is generated by  $w_0x^{q-1}+w_1x^{q-2}+f_{q-2,\theta}(x)$ with $w_0=0, w_1 \neq 0$ or $w_0 \neq 0$ and $\eta(w_0^2-4w_0w_1\theta)=-1$.
		
		\item [\textbf{(iii)}] When $k=q-3$, $\bm u$ is a deep hole of ${\rm TRS}_k(\mathbb{F}_q, \theta)$ if and only if $\bm u$ is given by  Theorem \ref{thm5} or generated by $a(x^{q-2}+\frac{1}{3\theta}x^{q-3})+f_{q-3,\theta}(x)$ with $a \neq 0$, $f_{q-3,\theta}(x) \in \mathcal{S}_{q-3, \theta}$ and $\eta(-3)=-1$.
  \end{itemize}
	\end{theorem}
 \begin{proof}
     \textbf{(i):} For $k=q-1$,  we have $\rho({\rm TRS}_k(\mathbb{F}_q, \theta))=q-k=1$. Thus every non-codeword is a deep hole. The conclusion can be easily verified.
		
		\textbf{(ii): }For $k=q-2$, by Proposition \ref{prop8}, $\bm u$ is a deep hole of ${\rm TRS}_k(\mathbb{F}_q, \theta)$ if and only if, for any $\alpha \in \mathbb{F}_q$, 
		\begin{equation}\label{3odd}
			\det\left(\begin{array}{cc}
				1  & w_0 \\
				\alpha-\theta \alpha^{2} &w_1\\
			\end{array}\right)=\theta w_0\alpha^2-w_0 \alpha +w_1 \neq0.
		\end{equation}	
		Hence, Eq. \eqref{3odd} holds $\Longleftrightarrow 
		w_0 =0,w_1 \neq 0$, or,  $w_0 \neq 0$ and 
		$\theta w_0x^2+w_0x +w_1=0$ has no roots in $\mathbb{F}_q$. One can easily verify that the latter is equivalent to $\eta(w^2_0-4w_1w_0\theta)=-1$ and the vector $\bm u_f$ generated by $f(x)=w_0x^{q-1}+w_1x^{q-2}+f_{q-2,\theta}(x)$ satisfies $H\cdot \bm u^\top_f=(w_0, w_1)^\top$. Thus such vectors are all deep holes of ${\rm TRS}_{q-2}(\mathbb{F}_q, \theta)$.
		
		\textbf{(iii): } For $k=q-3$, by Proposition \ref{prop8}, $\bm u$ is a deep hole of ${\rm TRS}_k(\mathbb{F}_q, \theta)$ if and only if for any $\alpha \neq \beta \in \mathbb{F}_q$,
		\[D(\alpha,\beta)=\frac{\det\left(\begin{array}{ccc}
				1  & 1 & w_0 \\
				\alpha & \beta &w_1\\
				\alpha^2-\theta \alpha^{3} &\beta^2-\theta \beta^{3} &w_2\\
			\end{array}\right)}{\beta-\alpha}\neq0,\]
  which is  equivalent to that  for any $Z_2\in\mathbb{F}_q^*$ and $Z_1\in\mathbb{F}_q$, 
   \begin{eqnarray}\label{8}
       g(Z_1,Z_2)&\triangleq& D(Z_1+Z_2,Z_1-Z_2) \nonumber\\
       &=&(\theta w_1-w_0+2\theta w_0Z_1)Z_2^2-\Big(2\theta w_0Z_1^3-(3\theta w_1+w_0)Z_1^2 \nonumber\\
       &&+2w_1Z_1-w_2\Big)\neq 0.
   \end{eqnarray}
   In other words, $\bm u$ is a deep hole of ${\rm TRS}_k(\mathbb{F}_q, \theta)$ if and only if ${\rm N}(g(Z_1,Z_2))={\rm N}(g(Z_1,0))$. 
   
   Next, we determine all such vectors $\bm w \in\mathbb{F}_q^3\setminus\{\bm 0\}$  such that $\bm u$ is a deep hole of ${\rm TRS}_k(\mathbb{F}_q, \theta)$.
		
		\emph{Case 1}: $w_0=0$ and $w_1=0$. Then $g(Z_1,Z_2)=w_2 \neq 0$, i.e.,  Eq. (\ref{8}) holds.
  
  \emph{Case 2}: $w_0=0$ and $w_1\neq 0$. Then $g(Z_1,Z_2)=\theta w_1(Z_2^2+3(Z_1-\frac{1}{3\theta})^2-b),$ where $b=\frac{w_1-3\theta w_2}{3\theta^2w_1}$. By Proposition \ref{prop4}, it has that 
  \begin{equation*}
      {\rm N}(g(Z_1,Z_2))=q+v(b)\eta(-3)=\begin{cases}
          1 & {\rm if~} w_1=3\theta w_2 \text{~and~} \eta (-3)=-1,\\
         2q-1 & {\rm if~} w_1=3\theta w_2 \text{~and~} \eta (-3)=1,\\
        q-\eta(-3) & {\rm otherwise}. 
      \end{cases}
  \end{equation*}  
 Moreover, it has that ${\rm N}(g(Z_1,0))=1$  when $w_1=3\theta w_2,$ and $ \eta (-3)=-1$ and ${\rm N}(g(Z_1,0))\leq 2$ in other cases.  Hence, $\bm w=w_1(0,1,\frac{1}{3\theta})$ for $w_1\neq 0$  is a desired vector.

 \emph{Case 3}: $w_0\neq 0$. We claim that $g(Z_1,Z_2)$ at least has a solution $(Z_1,Z_2)$ with $Z_2\neq 0$ in this case. 
 
 Suppose $b(Z_1)=(\theta w_1-w_0+2\theta w_0Z_1)(2\theta w_0Z_1^3-(3\theta w_1+w_0)Z_1^2+2w_1Z_1-w_2)$, $\mathcal{Z}_0=\{Z_1\in\mathbb{F}_q:b(Z_1)=0\}$, $\mathcal{Z}_+=\{Z_1\in\mathbb{F}_q\setminus \mathcal{Z}_0:\eta(b(Z_1))=1\}$, and $\mathcal{Z}_-=\{Z_1\in\mathbb{F}_q\setminus \mathcal{Z}_0:\eta(b(Z_1))=-1\}$. Note that ${\rm N}(b(Z_1))\leq \deg (b(Z_1))=4$, and $(\theta w_1-w_0+2\theta w_0Z_1)g(Z_1,Z_2)=((\theta w_1-w_0+2\theta w_0Z_1)Z_2)^2-b(Z_1)$, then it holds that \begin{itemize}
     \item [1)] $|\mathcal{Z}_0|={\rm N}(b(Z_1))\leq 4$;
     \item[2)] ${\rm N}(g(Z_1,Z_2))\geq 2|\mathcal{Z}_+|$, this is because that for any $Z_1\in \mathcal{Z}_+$ and $$g(Z_1,\frac{\mp\sqrt{b(Z_1)}}{\theta w_1-w_0+2\theta w_0Z_1})=0;$$
     \item [3)] $q=|\mathcal{Z}_0|+|\mathcal{Z}_+|+|\mathcal{Z}_-|$;
     \item[4)] $\sum_{Z_1\in\mathbb{F}_q}\eta(b(Z_1))=|\mathcal{Z}_+|-|\mathcal{Z}_-|=2|\mathcal{Z}_+|+|\mathcal{Z}_0|-q\leq 2|\mathcal{Z}_+|+4-q$.
 \end{itemize}
 If $b(Z_1)$ is a  square of some polynomial in $\mathbb{F}_q[Z_1]$, that is, $|\mathcal{Z}_-|=0$, then $${\rm N}(g(Z_1,Z_2))\geq 2|\mathcal{Z}_+|\geq 2(q-4)>4\geq {\rm N}(g(Z_1,0)).$$
If  $b(Z_1)$ isn't a  square of some polynomial in $\mathbb{F}_q[Z_1]$, then it follows from  Proposition \ref{prop3} and $q>16$ that $2|\mathcal{Z}_+|\geq q-4-3\sqrt{q}=(\sqrt{q}-4)(\sqrt{q}+1)>0$.  Hence, there exists $Z_1\in \mathcal{Z}_+$ such that $Z_2=\frac{\mp\sqrt{b(Z_1)}}{\theta w_1-w_0+2\theta w_0Z_1}\neq 0$ and $g(Z_1,Z_2)=0$. 

In a word, $\bm u$ is a deep hole of ${\rm TRS}_k(\mathbb{F}_q, \theta)$ if and only if $\bm w=a(0, 0, 1)$ or $\bm w=a(0, 1, \frac{1}{3\theta})$ when $\eta(-3)=-1$, where $a \in \mathbb{F}^*_q$. The conclusion then follows from Proposition \ref{prop7}. 
  
 \end{proof}

\section{Conclusion}\label{sec5}
	
Due to its crucial role in decoding, determining the deep holes of Reed-Solomon codes is of great significance. As a generalization of Reed-Solomon codes, twisted Reed-Solomon codes have received increasing attention from scholars in recent years and have important applications in coding theory and cryptography. In this paper, we focus on the deep hole problem of twisted RS codes. For a general evaluation set $\mathcal{A} \subseteq \mathbb{F}_q$, we determine the covering radius and obtain a standard class of deep holes of twisted RS codes ${\rm TRS}_k(\mathcal{A}, \theta)$. Furthermore, we investigate the classification of deep holes of the full-length twisted RS code ${\rm TRS}_k(\mathbb{F}_q, \theta)$. We prove that there are no other deep holes of ${\rm TRS}_k(\mathbb{F}_q, \theta)$ for $\frac{3q-4}{4} \leq k\leq q-4$ when $q\geq 8$ is even, and $\frac{3q+3\sqrt{q}-7}{4} \leq k\leq q-4$ when $q>16$ is odd. We also completely determine their deep holes for $q-3 \leq k \leq q-1$. 
For odd $q$, our results have a narrower range of values for the dimension $k$ compared to the case where $q$ is even. One reason for this is that in the proof of Lemma \ref{lem9}, we utilized the following estimate:
\[\left|\sum_{W \in \mathbb{F}^*_q}\eta(Wf(W))\right|\leq 3\sqrt{q}.\]
Improving the estimate of the absolute value of this character sum will yield better results for odd $q$. 

Although the deep hole problem for TRS codes is similar to that of RS codes, the former seems to be more challenging. In this paper, we not only utilized many results such as character sums (Gauss sums) and the number of solutions to equations over finite fields, but we also observed that even completely determining all deep holes for some boundary cases (see Theorems \ref{thm7} and \ref{thm10}) is far from straightforward. This paper is a preliminary exploration of the deep hole problem of TRS codes. The main results focus on the situation that the dimension $k$ is relatively large (close to $q$). Investigating the case of $k$ is relatively small might be an interesting problem for future work. Furthermore, based on our results and Conjecture \ref{conj1} for RS codes, we propose the following conjecture on the completeness of deep holes for TRS codes.
	
\begin{conjecture}
For $2 \leq k \leq q-4$, all deep holes of ${\rm TRS}_k(\mathbb{F}_q, \theta)$ have generating polynomials $f(x)=ax^{k}+f_{k,\theta}(x)$ with $a \neq 0$ and $f_{k,\theta}(x) \in \mathcal{S}_{k,\theta}$.
\end{conjecture}

It is also interesting to study the deep holes of ${\rm TRS}_k(\mathcal{A}, \theta)$ for more general evaluation set $\mathcal{A} \subseteq \mathbb{F}_q$ and other classes of twisted Reed-Solomon codes.

\section*{Acknowledgments}
The work of Weijun Fang was supported in part by the National Key Research and Development Program of China under Grant Nos. 2021YFA1001000 and 2022YFA1004900, the National Natural Science Foundation of China under Grant No. 62201322, and the Shandong Provincial Natural Science
Foundation under Grant No. ZR2022QA031, and the Taishan Scholar Program of Shandong Province. The work of Jingke Xu is supported in part by the National Natural Science Foundation of China under Grant No. 12201362,
and the Shandong Provincial Natural Science
Foundation under Grant No. ZR2021QA043.

\appendix

\section{ Proof of Proposition \ref{prop10}}\label{A}

Before proving Proposition \ref{prop10}, we need some notions and auxiliary lemmas. 
Suppose $q=2^m$, and  $\chi$ is the canonical additive character of $\mathbb{F}_q$. Let $\lambda$ be the multiplicative character of $\mathbb{F}_q$ of order 3 when $m$ is even,
and  $\pi$ be a primitive element of $\mathbb{F}_q$.

\begin{lemma}\label{A10}
For $a \in \mathbb{F}^*_q$, 
\[\sum_{X \in \mathbb{F}_q}\chi(aX^3)=
\begin{cases}
0  & \text { if } m \text { is odd},\\
(\Bar{\lambda}(a)+\Bar{\lambda}^2(a))\sqrt{q}  & \text { if } \frac{m}{2} \text { is odd},\\ 
-(\Bar{\lambda}(a)+\Bar{\lambda}^2(a))\sqrt{q} &\text { if } \frac{m}{2} \text { is even}.
\end{cases}\]
\end{lemma}

\begin{proof}
    By \cite{SE11}, it holds that $$G(\lambda,\chi)=\begin{cases}
     \sqrt{q} &\text { if } \frac{m}{2} \text { is odd},\\ 
-\sqrt{q} &\text { if } \frac{m}{2} \text { is even}. 
    \end{cases}
$$
By \cite[Theorem 5.30]{LN97}, we obtain that for any $a \in \mathbb{F}^*_q$,
\[\sum_{X \in \mathbb{F}_q}\chi(aX^3)=\begin{cases}
    0  & \text { if } m \text { is odd},\\ 
\Bar{\lambda}(a)G(\lambda,\chi)+\Bar{\lambda}^2(a)G(\lambda^2,\chi) &\text { if } m \text { is even}.
\end{cases}\]
Then the conclusion follows from the above two identities.   
 \end{proof}

\begin{lemma}\label{A11}
   Suppose $q=2^{m}$ and $a \in \mathbb{F}^*_q$, then 
   \[{\rm N}(XY(X+Y)+a)=\begin{cases}
       q-2 &\text { if } m \text { is odd},\\ 
q-2+2\sqrt{q} &\text { if } \frac{m}{2} \text { is odd and } \lambda(a)=1,\\ 
q-2+\sqrt{q} &\text { if } \frac{m}{2} \text { is even and } \lambda(a)\neq 1.
   \end{cases}
\]
\end{lemma}
\begin{proof}
\begin{eqnarray*}
	{\rm N}(XY(X+Y)+a)&=& \frac{1}{q}\sum_{X, Y \in \mathbb{F}_q}\sum_{z \in \mathbb{F}_q}\chi(z(XY(X+Y)+a)) \\
	&=&q-1+\frac{1}{q}\sum_{z \in \mathbb{F}^*_q}\sum_{Y \in \mathbb{F}^*_q}\sum_{X \in \mathbb{F}_q}\chi(zY X^2+z Y^2X+za) \\
	&\stackrel{\rm(a)}{=}&q-1+\sum_{Y\in \mathbb{F}^*_q}\chi(\frac{a}{Y^3})=q-1+\sum_{Y\in \mathbb{F}^*_q}\chi(aY^3)\\
 &=&q-2+\sum_{Y\in \mathbb{F}_q}\chi(aY^3),\\
\end{eqnarray*}
where Eq. (a) follows from  Proposition \ref{prop2} (iii).
   The conclusion then follows from Lemma \ref{A10}  and the fact that $\Bar{\lambda}(a)+\Bar{\lambda}^2(a)+1=0$ when $\lambda(a)\neq 1$
\end{proof}

\begin{lemma}[\cite{LN97}, Example 6.38]\label{A12}
    For any $b \in \mathbb{F}_q$, ${\rm N}(X^3+Y^3-b) \geq q-2\sqrt{q}-2.$ Consequently, when $q \geq 16$, there exist $x, y \in \mathbb{F}^*_q$, such that $b=x^3+y^3$.
\end{lemma}

\begin{lemma}\label{A13}
Suppose $q=2^m \geq 16$, and  $1 \leq n \leq \frac{q}{4}$. Let $\pi$ be a primitive element of $\mathbb{F}_q$. For $w\in \mathbb{F}_q$, denote $Q(x_1, \cdots, x_{n})=\sum\limits_{i=1}^{n}x^3_i+\theta^{-3}+(\sum\limits_{i=1}^{n}x_i+\theta^{-1})^3+\theta^{-1}w$.  Then there exist pairwise distinct elements $\alpha_1, \cdots, \alpha_{n} \in \mathbb{F}_q$ and $c \in \mathbb{F}_q^*$, such that 
\begin{equation*}
    Q(\alpha_1, \cdots, \alpha_{n})=g(c)\triangleq\begin{cases}c &{\rm if~} 2\nmid m,\\
    c^3 & {\rm if~} 2 \mid m \text{ and } 2\nmid \frac{m}{2},\\
    \pi c^3 & {\rm if~} 2 \mid m \text{ and } 2\mid \frac{m}{2}.\end{cases}\end{equation*} 
\end{lemma}
\begin{proof}
For $n=1$, note that $Q(x_1)=\theta^{-1}x_1(x_1+\theta^{-1})+\theta^{-1}w$. Then ${\rm N}\big(Q(x_1)+g(c)\big)>0$ for some $c\in\mathbb{F}_q^*$ if and only if ${\rm Tr}\big(\theta^2w+\theta^3 g(x)\big)=0$ has a solution $x \in \mathbb{F}_q^*$.  This is equivalent to  $1<{\rm N}\big({\rm Tr}(\theta^3 g(x))\big)<q$. Since ${\rm N}\big({\rm Tr}(\theta^3 g(x))+1\big)+{\rm N}\big({\rm Tr}(\theta^3 g(x))\big)=q$ and $\sum\limits_{x\in\mathbb{F}_q}\chi(\theta^3 g(x))={\rm N}\big({\rm Tr}(\theta^3 g(x))\big)-{\rm N}\big({\rm Tr}(\theta^3 g(x))+1\big)$, hence ${\rm N}\big({\rm Tr}(\theta^3 g(x))\big)=\frac{q}{2}+\frac{1}{2}\sum_{x\in\mathbb{F}_q}\chi(\theta^3 g(x))$. By Lemma \ref{A10}, it has that 
$|{\rm N}\big({\rm Tr}(\theta^3 g(x))\big)-\frac{q}{2}|\leq \sqrt{q}$, which implies that $$2<\frac{q}{2}-\sqrt{q}\leq {\rm N}\big({\rm Tr}(\theta^3 g(x))\big)\leq \frac{q}{2}+\sqrt{q}<q.$$
When $n\geq 2$, $Q(x_1, \cdots, x_{n-1},x_n=0)=Q(x_1, \cdots, x_{n-1})$, thus we may assume that $n$ is even and prove that there exist pairwise distinct elements $\alpha_1, \cdots, \alpha_{n} \in \mathbb{F}^*_q$ such that the conclusion holds. Calculating yields the recursion relation 
 \[Q(x_1, \cdots, x_{n})=(x_{n-1}+\sum_{i=1}^{n-2}x_i+\theta^{-1})(x_n+\sum_{i=1}^{n-2}x_i+\theta^{-1})(x_{n-1}+x_n)+Q(x_1, \cdots, x_{n-2}).\]
 Thus we will prove the conclusion by induction on $n\geq 2$. 
 
Case (i): $2 \nmid m$. For $n=2$, then $Q(x_1, x_2)=(x_1+\theta^{-1})(x_2+\theta^{-1})(x_1+x_2)+\theta^{-1}wXY(X+Y)+\theta^{-1}w$, where $X=x_1+\theta^{-1}$ and $Y=x_2+\theta^{-1}$. Choose $c \in \mathbb{F}_q^*\backslash\{\theta^{-1}w_{r-1}\}$, by Lemma \ref{A11}, ${\rm N}(Q(x_1, x_2)+c) ={\rm N}(XY(X+Y)+\theta^{-1}w+c) =q-2$. Note that $Q(x_1, x_2)+c=0$ implies that $x_1 \neq x_2$. In addition, ${\rm N}(Q(0, x_2)+c) \leq 2$ and ${\rm N}(Q(x_1, 0)+c) \leq 2$. Then from $q-2 >2+2=4$, we deduce that there exist $\alpha_1 \neq \alpha_2 \in \mathbb{F}^*_q$ such that $Q(\alpha_1, \alpha_2)+c=0$. Now suppose that there exist pairwise distinct elements $\alpha_1, \cdots, \alpha_{n-2} \in \mathbb{F}^*_q$ such that $Q(\alpha_1, \cdots, \alpha_{n-2})=c'$ for some $c' \in \mathbb{F}_q^*$. Denote $X=x_{n-1}+\sum_{i=1}^{n-2}\alpha_i+\theta^{-1}$ and $Y=x_{n}+\sum_{i=1}^{n-2}\alpha_i+\theta^{-1}$, then
\[Q(\alpha_1, \cdots,\alpha_{n-2}, x_{n-1,} x_{n})=XY(X+Y)+c'.\]
 Choose $c \in \mathbb{F}_q^*\backslash\{c'\}$, then by Lemma \ref{A11} again, ${\rm N}(Q(\alpha_1, \cdots,\alpha_{n-2}, x_{n-1,} x_{n})+c)={\rm N}(XY(X+Y)+c'+c)= q-2.$ Since $c \neq c'$, $Q(\alpha_1, \cdots,\alpha_{n-2}, x_{n-1,} x_{n})=c$ implies that $x_{n-1} \neq x_n$. Given $x_{n-1}$ (resp. $x_n$) $\in \{0, \alpha_1, \cdots, \alpha_{n-2}\}$, there exist at most two $x_n$'s (resp. $x_{n-1}$'s), such that  $XY(X+Y)+c'+c=0$. Thus, from $q-2 > 4(n-1)$, we obtain that there exist $x_{n-1}\neq x_{n} \in \mathbb{F}_q^*\backslash\{\alpha_1, \cdots, \alpha_{n-2}\}$  such that $Q(\alpha_1, \cdots, \alpha_{n})=\alpha$. 

Case (ii): $2 \mid m$ and $2 \nmid \frac{m}{2}$. By Lemma \ref{A12}, $\theta^{-1}w_{r-1}=c_1^3+d_1^3$ for some $c_1, d_1\in \mathbb{F}_q^*$. For $n=2$, then $Q(x_1, x_2)=XY(X+Y)+\theta^{-1}w_{r-1}=XY(X+Y)+c_1^3+d_1^3$, where $X=x_1+\theta^{-1}$ and $Y=x_2+\theta^{-1}$. Then by Lemma \ref{A11}, ${\rm N}(Q(x_1, x_2)+c_1^3)={\rm N}(XY(X+Y)+\beta_1^3)= q-2+2\sqrt{q}.$ Similar to the proof of (i), by $q-2+2\sqrt{q} > 4$, we deduce that there exist $\alpha_1 \neq \alpha_2 \in \mathbb{F}^*_q$ such that $Q(\alpha_1, \alpha_2)=c_1^3$. Now suppose that there exist pairwise distinct elements $\alpha_1, \cdots, \alpha_{n-2} \in \mathbb{F}^*_q$ such that $Q(\alpha_1, \cdots, \alpha_{n-2})=c'^3$ for some $c' \in \mathbb{F}_q^*$. By Lemma \ref{A12} again, $c'^3=c^3+d^3$ for some $c, d\in \mathbb{F}_q^*$. Denote $X=x_{n-1}+\sum_{i=1}^{n-2}\alpha_i+\theta^{-1}$ and $Y=x_{n}+\sum_{i=1}^{n-2}\alpha_i+\theta^{-1}$, then
\[Q(\alpha_1, \cdots,\alpha_{n-2}, x_{n-1,} x_{n})=XY(X+Y)+c'^3.\]
By Lemma \ref{A11} again, ${\rm N}(Q(\alpha_1, \cdots,\alpha_{n-2}, x_{n-1,} x_{n})+c^3)={\rm N}(XY(X+Y)+d^3)= q-2+2\sqrt{q}.$ Since $d \neq 0$, $X+Y\neq 0$, which implies that $x_{n-1} \neq x_n$.  Given $x_{n-1}$ (resp. $x_n$) $\in \{0, \alpha_1, \cdots, \alpha_{n-2}\}$, there exist at most two $x_n$'s (resp. $x_{n-1}$'s), such that $Q(\alpha_1, \cdots,\alpha_{n-2}, x_{n-1,} x_{n})+c^3=0$. Thus from $q-1+2\sqrt{q} > 4(n-1)$, we obtain that there exist $\alpha_{n-1}\neq \alpha_{n} \in \mathbb{F}_q^*\backslash\{\alpha_1, \cdots, \alpha_{n-2}\}$  such that $Q(\alpha_1, \cdots, \alpha_{n})+c^3=0$.

Case (iii) $4 \mid m$. By Lemma \ref{A12}, $\pi^{-1}\theta^{-1}w_{r-1}=c_1^3+d_1^3$ i.e., $\theta^{-1}w_{r-1}=\pi c_1^3+\pi d_1^3$ for some $c_1, d_1\in \mathbb{F}_q^*$. Note that $\lambda(\pi c_1^3)\neq 1$. The remainder proof is then similar to the case (ii) by using Lemmas \ref{A11} and \ref{A12}.
\end{proof}

Now, we are going to prove Proposition \ref{prop10}.

\textbf{Proof of Proposition \ref{prop10}}: Without loss of generality, we suppose $w_{r-3}=1$, then $w_{r-2}=\theta^{-1}$.
    By the proof process of Lemma \ref{lem4} (or the conclusion of Lemma \ref{lem4}), we have 
    \begin{eqnarray*}
        &&\frac{\det\big(\bm c_{r, \theta}(x_1)|\bm c_{r, \theta}(x_2)|\cdots|\bm c_{r, \theta}(x_{r-1})|\bm w^\top\big)}{\theta V(x_1, x_2, \cdots, x_{r-1})} \\
        &=& S_{r-3}(x_1, \cdots, x_{r-1})(\theta^{-1}+\sum\limits_{t=1}^{r-1}x_t)+S_{r-4}(x_1, \cdots, x_{r-1})\\
        &&+\theta^{-1}(S_{r-2}(x_1, \cdots, x_{r-1})(\theta^{-1}+\sum\limits_{t=1}^{r-1}x_t)+S_{r-3}(x_1, \cdots, x_{r-1})\big)+\theta^{-1}w_{r-1}\\
        &=&\sum_{1\leq i<j<\ell\leq r-1}x_ix_jx_{\ell}+\sum_{1\leq i<j\leq r-1}x_ix_j\sum_{i=1}^{r-1}x_i+\theta^{-1}(\sum_{i=1}^{r-1}x_i)^2+\theta^{-2}\sum_{i=1}^{r-1}x_i+\theta^{-1}w_{r-1}\\
        &=&Q(x_1, \cdots, x_{r-1}),
    \end{eqnarray*}
where $Q(x_1, \cdots, x_{r-1})$ is given as in Lemma \ref{A13}. Since $\frac{3q-4}{4}\leq k\leq q-4$, $1\leq r-3\leq \frac{q-8}{4}<\frac{q}{4}$. By Lemma \ref{A13}, there exist pairwise distinct elements $\alpha_1, \cdots, \alpha_{r-3} \in \mathbb{F}_q$ and $c \in \mathbb{F}_q^*$, such that \begin{equation*}
    Q(\alpha_1, \cdots, \alpha_{r-3})=g(c),\end{equation*} 
    where $g(c)$ is defined as in Lemma \ref{A13}.
 Note that  
\begin{eqnarray*}
   && Q(\alpha_1, \cdots, \alpha_{r-3}, x_{r-2}, x_{r-1})\\
   &=&(x_{r-2}+\sum_{i=1}^{r-3}\alpha_i+\theta^{-1})\cdot(x_{r-1}+\sum_{i=1}^{r-3}\alpha_i+\theta^{-1})
 \cdot(x_{r-2}+x_{r-1})+Q(\alpha_1, \cdots, \alpha_{r-3})\\
    &=&XY(X+Y)+g(c),
\end{eqnarray*}
where $X=x_{r-2}+\sum_{i=1}^{r-3}\alpha_i+\theta^{-1}$ and $Y=x_{r-1}+\sum_{i=1}^{r-3}\alpha_i+\theta^{-1}$. By Lemma \ref{A11}, ${\rm N}(Q(\alpha_1, \cdots, \alpha_{r-3}, x_{r-2}, x_{r-1}))={\rm N}(XY(X+Y)+g(c))\geq  q-2.$ Since $g(c) \neq 0$, $Q(\alpha_1, \cdots,\alpha_{n-2}, x_{r-2,} x_{r-1})=0$ implies that $x_{n-1} \neq x_n$.  Given $x_{r-2}$ (resp. $x_{r-1}$) $\in \{ \alpha_1, \cdots, \alpha_{r-3}\}$, there exist at most two $x_{r-1}$'s (resp. $x_{r-2}$'s), such that $Q(\alpha_1, \cdots, \alpha_{r-3}, x_{r-2}, x_{r-1})=0$.  Thus, from $q-2 > 4(r-3)$, we obtain that there exist $\alpha_{r-2}\neq \alpha_{r-1} \in \mathbb{F}_q^*\backslash\{\alpha_1, \cdots, \alpha_{r-3}\}$  such that $Q(\alpha_1, \cdots, \alpha_{r-1})=0$, i.e., there exist pairwise distinct elements $\alpha_1, \cdots,  \alpha_{r-1} \in \mathbb{F}_q$, such that $\det\big(\bm c_{r, \theta}(\alpha_1)|\bm c_{r, \theta}(\alpha_2)|\cdots|\bm c_{r, \theta}(\alpha_{r-1})|\bm w^\top\big)=0$. By Proposition \ref{prop8}, $\bm u$ is not a deep hole of ${\rm TRS}_k(\mathbb{F}_q, \theta)$.

\end{document}